\documentclass[format=acmsmall, review=false]{acmart}
\usepackage{acm-ec-26}
\usepackage{booktabs} % For formal tables
\usepackage[ruled]{algorithm2e} % For algorithms

\SetAlFnt{\small}
\SetAlCapFnt{\small}
\SetAlCapNameFnt{\small}
\SetAlCapHSkip{0pt}
\IncMargin{-\parindent}

\usepackage{multirow}

\usepackage{pgfplots}
\usepackage{subcaption}
\pgfplotsset{compat=1.18}

\definecolor{acmDarkBlue}{RGB}{0,114,178}
\definecolor{acmGreen}{RGB}{0,158,115}
\definecolor{acmPink}{RGB}{204,121,167}
\definecolor{acmOrange}{RGB}{213,94,0}
\definecolor{acmYellow}{RGB}{240,228,66}
\definecolor{acmLightBlue}{RGB}{86,180,233}

\newcommand{\new}{\textcolor{black}}
\newcommand{\nd}{\textcolor{black}}

\pgfplotsset{
    cycle list={
        {acmDarkBlue, mark=*},
        {acmGreen, mark=square*},
        {acmPink, mark=triangle*},
        {acmOrange, mark=diamond*},
        {acmYellow, mark=o},
        {acmLightBlue, mark=star}
    }
}

\usepackage{thm-restate}

\newcommand{\wpopOpt}{\textsc{max-weak-pm}}

\newcommand{\spopex}{\textsc{super-pm}}

\newcommand{\cstab}{$\gamma$-min stable}
\newcommand{\cgblocks}{$\gamma$-min blocks}
\newcommand{\cgstab}{$\gamma$-min stable}

\newcommand{\cgblocking}{$\gamma$-min blocking}

\newcommand{\cpopOPT}{\textsc{max-$\gamma$-pm}}

\newcommand{\spop}{super popular}
\newcommand{\wpop}{weakly popular}

\newcommand{\cpop}{$\gamma$-popular}

\newcommand{\votec}{\mathrm{vote}^{\gamma}}
\newcommand{\vote}{\mathrm{vote}}

\newcommand{\Dvote}{\Delta}
\newcommand{\Dvotec}{\Delta^{\gamma}}

\newcommand{\votew}{\mathrm{vote}^{weak}}
\newcommand{\votes}{\mathrm{vote}^{super}}
\newcommand{\Dvotes}{\Delta^{super}}
\newcommand{\Dvotew}{\Delta^{weak}}
\newcommand{\maxsmti}{\textsc{max-smti}}

\newcommand{\smti}{\textsc{max-smti}}

\usepackage{boxedminipage}
\usepackage{xspace}

\newtheorem{theorem}{Theorem}

\newtheorem{claim}{Claim}
\newtheorem{prop}{Proposition}
\newtheorem{example}[theorem]{Example}

\newenvironment{claimproof}{\par\noindent\underline{Proof:}}{\leavevmode\unskip\penalty9999 \hbox{}\nobreak\hfill\quad\hbox{$\blacksquare$}}

\newcommand{\pbDef}[3]{%
\noindent
\begin{center}
\begin{boxedminipage}{0.98\columnwidth}
#1\\[5pt]
\begin{tabular}{p{0.14\columnwidth}p{0.8\columnwidth}}
Input: & #2\\
Output: & #3
\end{tabular}
\end{boxedminipage}
\end{center}
}

\setcitestyle{authoryear}

\title[Weakly Popular and Super Popular Matchings]{Weakly Popular and Super Popular Matchings}

\author{Gergely Cs\'{a}ji}
\orcid{0000-0002-3811-4332}
\affiliation{%
\department{Institute of Economics}
\institution{ELTE KRTK}
\city{Budapest}
\country{Hungary}
}
\author{Frederik Glitzner}
\orcid{0009-0002-2815-6368}
\affiliation{%
\department{School of Computing Science}
\institution{University of Glasgow}
\city{Glasgow}
\country{United Kingdom}
}
\begin{abstract}
\new{The efficient computation of large matchings with desirable fairness or stability guarantees is a crucial objective in centralised matching scheme design. However, even in simple two-sided matching markets with weak ordinal preferences, finding a maximum-size stable matching is NP-hard. Alternatively, popular matchings \nd{can be of larger size and} are weak Condorcet winners, but their existence is not guaranteed.} In this paper, we study a \new{novel} definition of popularity in bipartite graphs with two-sided preferences, \new{possibly with ties}, \nd{where agents are only indifferent between two matchings if they receive the same partner.} We show that this alternative definition of popularity, which we call weak popularity, \new{guarantees} the existence of such matchings, \nd{which can be significantly larger than stable matchings. Unfortunately, finding a maximum-size weakly popular matching turns out to be NP-hard even with one-sided ties. However, we provide} a polynomial-time algorithm to find a weakly popular matching that has at least $\frac{3}{4}$ times the size of a maximum-size weakly popular matching. We also show that this matching is at least $\frac{4}{5}$ times the size of a maximum-size weakly stable matching. \nd{Furthermore, we highlight that \new{our} $\frac{3}{4}$-approximation is best possible (under some complexity-theoretic assumptions). }

\new{We} complement our approximation results with an Integer Linear Programming formulation that solves the maximum-size weakly popular matching problem exactly. We extensively evaluate our algorithms on both randomly generated and real-world instances. Our experiments demonstrate that weakly popular matchings can be significantly larger than stable matchings, often covering all agents. Furthermore, we show that our approximation algorithm performs nearly optimally in practice. \new{Finally, we show that maximum-size weakly popular matchings can have very few blocking edges, suggesting that weak popularity offers a desirable trade-off between size and stability.}

We also study a model more general than weak popularity, where for each edge, we can specify for both agents the size of improvement the agent needs to vote in favor of a new matching. We show that even in this more general model, a so-called $\gamma$-popular matching always exists, and \new{our approximation algorithm applies.} Finally, we define a stronger variant of popularity, called super popularity, where even a weak improvement is enough to vote in favor of a new matching. We show that for this case, even the existence problem is NP-hard.
\end{abstract}

\begin{document}

% Title page for title and abstract only.
\begin{titlepage}

\maketitle
\makeatletter
\gdef\@ACM@checkaffil{}
\makeatother
% Optionally include a table of contents
\setcounter{tocdepth}{1} % adjust to 1 if desired
\tableofcontents

\end{titlepage}

% Paper 
\section{Introduction}

\subsection{Motivation}

In \nd{ordinal} preference-based \new{matching} markets, each market participant \nd{(agent)} expresses their preferences as a strict or weak ordering of some possible scenarios, and the aim is to find solutions that are fair or optimal according to some fairness or optimality notion. The area has been extensively studied in economics, mathematics, and computer science, originating from the seminal paper of \citet{gale1962college} about stable matchings. \nd{Here,} stability means that we aim for solutions \new{without} pairs of agents that could mutually improve by dropping their \new{current} partners and getting matched together \new{instead}. Their model for the stable matching problem has been applied to numerous real-world scenarios, including resident allocation, school choice, and even kidney exchange programs. For an extensive overview of the area of \new{matching under preferences, we refer to} \citet{manlove2013algorithmics}.

A relaxation of stability, called popularity, which gained significant theoretical interest in the literature recently, was introduced by \citet{gardenfors1975match}. Popular matchings translate the simple majority voting rule into the world of matchings under preferences, i.e., a matching is popular if it does not lose a head-to-head comparison to any other matching\new{, where an agent's vote is derived only from the partners they receive in the two matchings. This makes popular matchings equivalent to weak Condorcet winners \cite{CsehKavitha21}.}

Popular matchings have several advantages compared to stable matchings. First of all, a popular matching can be twice the size of a (maximum-size) stable matching \cite{huang2013popular}, so whenever matching many agents is an essential feature in an application, popularity can be a more favorable \new{objective} than stability. While \new{a popular matching} can admit blocking edges, it still satisfies a nice form of "global stability", meaning \nd{that} it cannot be blocked by any other matching. Another nice feature of popular matchings is that they can be defined for both one-sided and two-sided preferences, while stable matchings can only be defined for two-sided ones. One-sided markets are just as important from a practical point of view as two-sided ones, with applications ranging from fair allocation \cite{suksompong2021constraints} to refugee settlement \cite{refugeematching} and student allocation \cite{studentallocation}. Thirdly, popularity is stronger than Pareto-optimality, which is one of the most used optimality notions in practice \cite{cseh2017popular}. Finally, another nice feature of popular matchings is that their definition can be naturally made compatible with requiring maximum-size or maximum-weight solutions by restricting the possible candidate matchings. Hence, they are also useful in applications where maximizing some utility of the agents, or of a central authority, has a higher priority than stability. For an extensive overview of popular matchings, see \citet{cseh2017popular}.

However, a large drawback of popular matchings is that \new{the existence of popular matchings is not guaranteed when preference lists are not strictly ordered. Even} worse, deciding if an instance with ties admits a popular matching is NP-complete, even if ties occur on one side of the market only, as \citet{cseh2017onesidedties} showed. In contrast, \new{(weakly)} stable matchings always exist in this model; however, a maximum-size stable matching \nd{is} NP-hard to find \cite{IMMMmaxsmti}, and the current best approximation ratio is $\frac{2}{3}$ \cite{Mcdermid09,kiraly2012linear}, which holds even in significantly more general settings \cite{csaji2023simple}.

\subsection{Our contributions}
\label{sec:contribution}

\nd{Due to the absence of existence guarantees of popular matchings in the presence of ties, we propose and study a new notion of popularity called \emph{weak popularity}, which differs from classical popularity as follows: any agent who compares a current matching $M$ to an alternative matching $N$ prefers to stick to $M$ (rather than being indifferent between them) when they receive an equally good (but not the same) partner in $N$ compared to $M$. A matching is then weakly popular if it does not lose in a head-to-head election against any other matching given this voting rule. This notion is inspired by the classical notion of weak stability, in which a pair of agents blocks only if they are both \emph{strictly} better off being assigned to each other. Weak popularity is practically motivated: suppose that a matching procedure has already taken place, e.g., students are already assigned to courses or resident doctors are already assigned to hospitals. Then, unless agents get better off, they would likely prefer to stick to the current matching as they may have already started the course or made arrangements. We also generalise this notion of weak popularity to \emph{$\gamma$-popularity} (inspired by $\gamma$-min stability), in which agents provide thresholds for the increase in utility they require from an alternative matching $N$ to vote for $N$ rather than the current $M$.}

\nd{We prove that, unlike classical popular matchings, weakly popular and $\gamma$-popular matchings always exist and can be found in polynomial time. However, we show that they can have different sizes and that finding maximum-size weakly popular (and by generalisation also $\gamma$-popular) matchings is NP-hard, even if ties occur on one side of the market only, ruling out an exact polynomial time algorithm for these problems (unless P=NP). However, on the positive side, we show that weak and $\gamma$-popularity can be verified efficiently and provide an algorithm to find approximately maximum-size such matchings in linear time. In particular, we show that in both the weak popularity model and the more general $\gamma$-popularity model, there exists an efficient $\frac{3}{4}$-approximation algorithm. We also show that the size of the output of this algorithm is at least $\frac{4}{5}$ times the size of a maximum-size stable matching (in the case of weak popularity) or of a maximum-size $\gamma$-min stable matching (in the case of $\gamma$-popularity), both of which are NP-hard to find (and $\frac{2}{3}+\varepsilon$ inapproximability results exist~\cite{smti1.5inapprox}). This brings a significant advantage in practice, where existence and size guarantees are crucial. Finally, we establish that our $\frac{3}{4}$ approximation ratio is tight under the strong-Unique Games Conjecture (or the Small Set Expansion Hypothesis). On the technical side, we use the edge duplication technique. While this technique is similar to the idea of promoting agents \cite{kiraly2012linear}, it allows us to handle, analyze, and solve problems much more easily. Our algorithm combines and extends the one of \citet{kiraly2012linear} for \smti\ (generalised in \cite{csaji2025simple} for the maximum-size \cgstab\ matching problem), and the one by \citet{kavitha2014size} for the maximum-size popular matching with strict preferences.}

\nd{To further bridge the gap between theory and practice, we also address the computation of exact solutions and provide empirical evidence for the advantageous performance of our algorithms. We provide an exact Integer Linear Programming (ILP) algorithm for finding maximum-size weakly popular matchings and substantiate our theoretical contributions with extensive experiments on both randomly generated instances and real-world data from the Scottish Foundation Allocation Scheme (SFAS).\footnote{See \url{https://matching-in-practice.com/the-scottish-foundation-allocation-scheme-sfas/}.} Our empirical results highlight three key findings: first, weakly popular matchings can be significantly larger than weakly stable matchings -- matching all or nearly all agents and, importantly, up to $10\%$ more agents than maximum-size weakly stable matchings. Second, we observe that our approximation algorithm performs near optimally in practice, often finding a maximum-size weakly popular matching. Third, by further minimising the number of blocking edges, we observe that the outcomes are remarkably close to being stable, suggesting that weak popularity can offer a notable size increase without introducing much instability.}

\nd{Finally, we introduce another alternative notion of popularity in which agents who vote between a current matching $M$ and an alternative matching $N$ would strictly prefer to switch to $N$ even when they receive an equally good (but not the same) partner in $N$ compared to $M$. We motivate this through altruistic agents who have nothing to lose, but could help other agents be strictly better off by voting for the alternative. We call any matching that does not lose in a head-to-head election against any other matching with this modified voting rule \emph{super popular}, inspired by the classical notion of super stability. Unfortunately, super-popular matchings might not exist and we show that deciding the existence is NP-hard, even if there are only two ties in the instance.}

\subsection{Structure of the paper}

\nd{In Section \ref{sec:prelim}, we discuss related work and provide formal definitions and preliminary observations. We then prove that finding maximum-sized weakly and $\gamma$-popular matchings is NP-hard and provide our central approximation algorithm in Section \ref{sec:approxalgo}. Section \ref{sec:verify} focuses on the verification and exact computation of (large) weakly and $\gamma$-popular matchings and the intractability of finding (any) super popular matchings. Finally, in Section \ref{sec:exp}, we provide and discuss our experimental results, before concluding in Section \ref{sec:conclusion}.
}

\section{Preliminaries}
\label{sec:prelim}

\subsection{Related work}
\label{sec:related} 

The Stable Marriage problem was introduced by \citet{gale1962college}, \new{who showed that a stable matching always exists and how one can be computed in linear time. Under strict preferences, all stable matchings have the same size \cite{galesoto85}, but this is not true for the setting with ties and incomplete lists. Specifically,} the \emph{stable marriage problem with ties and incomplete lists} (\textsc{smti}) was first studied by \citet{IMMMmaxsmti}, who showed the NP-hardness of \smti, which is the problem of finding a maximum-size stable matching in an \textsc{smti} instance. Super stability, which is a stronger version of stability, was introduced by \citet{irving1994stable}, who proved that a super stable matching can be found in polynomial time, if one exists. Also, interestingly, super stable matchings have the same size, even with weak preferences. 

Later, it was shown by \citet{Mcdermid09} that \smti\ can be approximated in polynomial time within a factor of $\frac{2}{3}$. \citet{kiraly2012linear} gave a much simpler linear-time algorithm with the same approximation ratio, using a multilevel Gale-Shapley algorithm. On the negative side, \citet{halldorsson2002inapproximability} showed that it is NP-hard to approximate \smti\ within some constant factor. \citet{yanagisawa2007approximation} and \citet{halldorsson2003improved} showed that, assuming the Unique Games Conjecture (UGC), there is no $\frac{3}{4}+\varepsilon$-approximation for any $\varepsilon >0$. Recently, \citet{smti1.5inapprox} proved that assuming the Small Set Expansion Hypothesis or the strong-UGC, there cannot even be a $\frac{2}{3}+\varepsilon$-approximation algorithm for \smti.

Popularity as a solution concept was first introduced for the two-sided preference model by \citet{gardenfors1975match}, who showed that every stable matching is popular. \citet{abraham2007popular} introduced popularity for the one-sided House Allocation model. Both of these papers provided polynomial-time algorithms for finding a popular matching in the given model. Later, the focus of research shifted to finding maximum-size popular matchings; \citet{huang2013popular} showed that maximum-size popular matchings can also be found in polynomial time. \citet{kavitha2014size} gave an algorithm for this problem that used a multilevel Gale-Shapley approach, which is similar to our edge-duplicating method, but the latter allows much easier analysis and is easier to generalise. 

We note that for the non-bipartite Stable Roommates model, \new{deciding the existence of} a popular matching is NP-hard \cite{faenza2019popular,gupta2021popular}. Similarly, if ties are allowed in the preferences of the agents, then even in the two-sided model, a popular matching may fail to exist, and finding one becomes NP-hard, as \citet{cseh2017onesidedties} showed. To overcome this obstacle, we introduce an alternative definition of popularity in this case.

Notions related to our $\gamma$-popularity have been studied in \cite{chen2021matchings,csaji2023approximation,csaji2023simple}. \citet{chen2021matchings} defined \emph{local $d$-near-stability} based on swaps in the preference orders, such that a matching is locally $d$-near stable if there is no blocking edge $(u,w)$ that still remains blocking even after moving $w$ and $u$ $d$ places down their original preference lists. In \citet{csaji2023approximation}, a similar notion for cardinal preferences called $\Delta$-min stability was introduced, which trquires that for any edge to block, both endpoints should improve by at least $\Delta$. \citet{csaji2023simple} generalised this to a concept called $\gamma$-min stability, where for each edge $e=(u,w)$, distinct $\gamma_e^u>0$ and $\gamma_e^w>0$ values are given, and for an edge $(u,w)$ to block, $u$ should improve by at least $\gamma_e^u$ and $w$ by $\gamma_e^w$. He extended the $\frac{2}{3}$-approximation algorithm for \smti\  to this more general setting. 

\new{Another notable alternative to stability and popularity in the absence of stable or popular maximum-cardinality matchings is the computation of ``almost-stable'' maximum-cardinality matchings, that is, maximum-cardinality matchings with a minimum number of blocking pairs (or, similarly, blocking agents). \citet{biro_sm_10} proved that such a matching is NP-hard to approximate within $n^{1-\varepsilon}$ (for $n$ agents and for any $\varepsilon>0$). \citet{hamada09} strengthened this result to the case where all preference lists are of length at most 3, \citet{gupta2020parameterized} first considered this problem from a parametrised complexity perspective, and \citet{chen2025fptapproximabilitystablematchingproblems} proved that it is {\sf W[1]-hard} with respect to the optimal value to even approximate the problem within any computable function depending only on this optimal value. Recently, \citet{minimax,glitznermanloveasaamas26} studied the problem from a minimax fairness perspective, proving that even deciding whether there exists a matching with no more than one blocking pair per agent is NP-complete, thus deeming almost-stable matchings an impractical alternative for many large real-world matching markets.}

\new{The Hospitals/Residents problem generalises the Stable Marriage problem by assuming that one side of the market (the hospitals) can be matched to more than one agent. Stability remains defined in a similar way, and we refer to \citet{manlove2013algorithmics} for an extensive overview. Popularity in this} setting has been studied, for example, by \citet{nasre2017popular,nasre2017popularity} and \citet{krishnapriya2018good}. \citet{nasre2017popular} showed that, even in the presence of lower quotas, whenever a feasible matching exists, there also exists a popular matching among the feasible ones, and they provided an efficient algorithm to find such a matching. \citet{nasre2017popularity} presented an algorithm for computing a maximum-size popular matching in a setting where residents are partitioned into classes, and hospitals have laminar upper quotas that bound the number of residents from each class.

The only work that investigates the relationship between stability and popularity from an experimental perspective appears to be by \citet{krishnapriya2018good}. Their study, however, is restricted to strict preferences on both sides. They observed that in this setting, popular matchings can be up to 10$\%$ larger than stable matchings, are often close in size to maximum matchings, and exhibit relatively few blocking edges. In contrast, we consider instances with ties, where finding a maximum-size stable matching is NP-hard and a popular matching may not exist (although weakly popular matchings do). Despite the more general setting, we observe similar behavior. We also consider minimizing blocking edges with respect to maximum-size popularity constraints to show that surprisingly few blocking edges are needed in that case, less than 1 on average.

\subsection{Formal definitions}

We investigate matching markets, where the set of agents with the possible set of contracts is given by a (not necessarily simple) bipartite graph $G=(U,W;E)$. For each agent $v\in U\cup W$, we denote by $E(v)$ the edges that are incident to $v$. We assume that for each agent $v$, there is a \textit{preference valuation} $p_v:E(v) \to \mathbb{R}_{\ge 0}$, which defines a weak ranking over the adjacent edges of $v$. This is a strictly more general setting than ordinary weak preference lists, where only a weak order $\succeq_v$ is given for each agent $v\in U\cup W$ over $E(v)$. We note that we assume that the vertices (i.e., the agents) rank their incident edges instead of the adjacent vertices on the other side, because we allow parallel edges in our model. This corresponds to allowing multiple types of contracts between two given agents. For example, in resident allocation, there may be multiple edges between a hospital and a resident corresponding to different salaries for the same position, and it may happen that even though the resident $r$ is worse than some other resident $r'$, a contract with a lower salary for $r$ is preferred by the hospital to the better resident $r'$ with a higher salary. We also assume that $p_v(\emptyset )=-\infty$, which means that an agent always strictly prefers to be matched to any acceptable partners rather than being unmatched.

We say that an edge set $M\subseteq E$ is a \textit{matching}, if $|M\cap E(v)|\le 1$ for each $v\in U\cup W$. For a vertex $v$, let $M(v)$ denote the edge incident to $v$ in $M$, if there is any, otherwise $M(v)=\emptyset$. A matching $M$ is \textit{maximal} if no edge can be added to it. A matching is \textit{maximum} if there exists no larger one.

\paragraph{\textbf{Stability.}}

We start by defining weakly stable matchings. \new{Given a matching $M$, we say that an edge $e=(u,w)\notin M$ is \emph{dominated} at $u$ (resp. $w$), if $p_u(M(u))\ge p_u(e)$ (resp. $p_w(M(w))\ge p_w(e)$). }
\new{We say that a matching $M$ is \textit{weakly stable}, or simply \textit{stable}, if every edge $e=(u,w)\notin M$ is dominated.} Otherwise, such an edge is called a \textit{blocking edge}. In the context of ordinary weak preference orders $\succeq_v$ for $v\in U\cup W$ this is just equivalent to saying that there should be no edge $(u,w)$ where $u$ and $w$ strictly prefer each other to their partners in $M$.  

\paragraph{\textbf{Super stability.}}

There is a stronger version of stability, called super stability, which is defined as follows. Let $M$ be a matching. We say that $M$ is \textit{super stable} if there is no edge $e=(u,w) \notin M$ such that $p_u(e)\ge p_u(M(u))$ and $p_w(e)\ge p_w(M(w))$, so here even a weak improvement is enough to block $M$. \nd{Clearly, a super stable matching need not exist.}

\paragraph{\textbf{$\gamma$-min stability.}}

Now we define $\gamma$-min stability.
Suppose we are given $\gamma_e^u>0,\gamma_e^w>0$ values for each edge $e=(u,w)\in E$. They specify the size of improvement an agent $u$ needs in order to prefer the edge $e$ to their current assignment. Let $M$ be a matching. We say that an edge $e=(u,w)\notin M$ \textit{$\gamma$-min blocks $M$}, if $p_u(e)\ge p_u(M(u))+ \gamma_e^u$ and $p_w(e)\ge p_w(M(w)) +\gamma_e^w$. We say that $M$ is \textit{$\gamma$-min stable}, if there is no $\gamma$-min blocking edge to $M$. The case of weak stability corresponds to the special case where every $\gamma_e^v$ value is sufficiently small, and super stability would correspond to the case where every $\gamma_e^v$ is $0$ \nd{(but it is necessary for our subsequent results to assume strictly positive values for $\gamma_e^v$ instead).} 

\paragraph{\textbf{Popularity.}} 

Now, we define popular matchings with their original definition. Let $M$ be a matching, and let $N$ be another matching. Each agent $v\in U\cup W$ compares these two matchings and casts a vote as follows. If $p_v(M(v))=p_v(N(v))$, then we define $\vote_v(M,N)=0$. If $p_v(M(v))>p_v(N(v))$, then $\vote_v(M,N)=+1$, and finally if $p_v(M(v))<p_v(N(v))$, then $\vote_v(M,N)=-1$. Then, we aggregate the votes of each agent and obtain $\Dvote (M,N)=\sum_{v\in U\cup W}\vote_v(M,N)$. Clearly, with this definition, we have $\vote_v(M,N)+\vote_v(N,M)=0$. 
We say that a matching $M$ is \emph{popular}, if for any matching $N$, it holds that $\Dvote (M,N)\ge 0$. \new{Else, if $\Dvote (M,N)<0$, we say that $N$ \emph{dominates $M$}.}
Furthermore, we say that $M$ is \emph{dominant}, if (i) $M$ is a maximum size popular matching and (ii) $\vote (M,N)>0$ for any matching $N$ with $|N|>|M|$. 

\paragraph{\textbf{Weak popularity.}}

We provide an alternative definition of popularity as follows. Let $M,N$ be two matchings. We define $\votew_v(M,N)$ in a slightly different manner. If $p_v(M(v))>p_v(N(v))$ or $p_v(M(v))<p_v(N(v))$, then $\votew_v(M,N)$ is $+1$ and $-1$ respectively, same as before. Also, if $M(v)=N(v)$, then again $\votew_v(M,N)=0$. However, when $M(v)\ne N(v)$, but $p_v(M(v))=p_v(N(v))$, we define $\votew_v(M,N)=+1$. \nd{Recall our motivation for this notion from Section \ref{sec:contribution}:} if an agent $v$ does not improve in the matching $N$ compared to $M$, but has to change partners, then they would probably prefer to keep their original partner and not go through the additional effort to obtain a new partner that they value equally. Then, we define a matching $M$ to be \textit{\wpop}, if for any matching $N$, it holds that $\Dvotew (M,N)=\sum_{v\in U\cup W}\votew_v(M,N)\ge 0$. 
%We define $M$ to be \textit{\wdom} if (i) $M$ is a maximum size \wpop\ matching and (ii) for any matching $N$ with $|N|>|M|$, it holds that $\votew(M,N)>0$. 

\paragraph{\textbf{$\gamma$-popularity.}}

We also generalise weak popularity, inspired by $\gamma$-min stability. Suppose again that we are given $\gamma_e^u>0,\gamma_e^w>0$ values for each edge $e=(u,w)\in E$. \nd{For any $v\in U\cup W$, we can assume that $\gamma_{\{v,v\}}^v=0$.} Now, let $M$ and $N$ be two matchings. We define $\votec_v(M,N)$ as follows. First, if $M(v)=N(v)$, then $\votec_v(M,N)=0$. \nd{Now, for $M(v)\ne N(v)$,} if $p_v(N(v))\ge p_v(M(v))+\gamma_{\nd{e}}^v$, \nd{where $e=N(v)$}, then $\votec_v(M,N)=-1$\nd{, and if} $p_v(N(v))<p_v(M(v))+\gamma_e^v$, \nd{where $e=N(v)$}, then $\votec_v(M,N)=+1$. Then, we define a matching $M$ to be \textit{\cpop}, if for any matching $N$, it holds that
$\Dvotec (M,N)=\sum_{v\in U\cup W}\votec (M,N)\ge 0.$
\nd{Again, recall our motivation for this notion from Section \ref{sec:contribution}: for any agent to deviate to a different partner,} the improvement they receive in the new matching must be large enough to justify the effort of deviation, and these efforts can be different for each agent and edge. In a course allocation setting, for example, if the students already started some courses according to a matching $M$, then even if a student would receive a better course in a different matching $N$, it may not be worth it to switch anymore, since they have to restart from the beginning. Hence, it is natural to assume that \new{there are certain valuation thresholds for the agents that need to be met} to consider a new matching worthy of changing to it. 

\paragraph{\textbf{Super popularity.}}

Motivated by super stable matchings, we also introduce a stronger variant of popularity, called super popularity. In the super stable matching model, even a partner that is equally good is considered good enough to block. Hence, now we define the votes as follows. 
Let $M,N$ be two matchings. If $p_v(M(v))> p_v(N(v))$, then $\votes_v(M,N)=+1$. If $N(v)=M(v)$, then $\votes_v(M,N)=0$. Finally, if $p_v(N(v))\ge p_v(M(v))$ and $M(v)\ne N(v)$, then $\votes_v(M,N)=-1$. We say that $M$ is \textit{\spop} if for any other matching $N$, it holds that $\Dvotes (M,N)=\sum_{v\in U\cup W}\votes_v(M,N)\ge 0.$ 
\nd{Recall from Section \ref{sec:contribution} that one} can think of this as having cooperative agents who are willing to help others deviate and receive better contracts by falsely reporting that they prefer the new matching, even if it is just as good for them. We could consider a more general model, where agents would help \new{others even if they were to receive a slightly worse partner in the new matching; however, we will see that even the super popularity objective above leads to strong computational intractability results.}

\subsection{Problem statements and observations}

Now we define the central optimisation problem that we investigate, called \textsc{Maximum \cpop\ matching with ties and incomplete preferences}, abbreviated as \cpopOPT. 
\pbDef{\cpopOPT}{
A bipartite graph $G=(U,W;E)$, $p_v()$ preference valuations for each $v\in U\cup W$, numbers $0<\gamma_e^v$ for each pair $(e,v)\in E\times (U\cup W)$ such that $v\in e$.
}{
A maximum-size \cpop\ matching $M$.
}

\nd{For the restricted case of weak popularity, which corresponds to the case when each $\gamma_e^v$ is sufficiently small such that any strict improvement is enough to change the vote}, we call the optimisation problem \wpopOpt.

For super popularity, we define the following existence problem. 
\pbDef{\spopex}{
A bipartite graph $G=(U,W;E)$ and $p_v()$ preference valuations for each $v\in U\cup W$.
}{
Is there a \spop\ matching $M$?
}
%The restriction of the problem to ordinary weak preferences is denoted by \spopex.

Before we move on to our main results, let us show a simple theorem that connects \cpop\ and \cgstab\ matchings.
\begin{theorem}
\label{thm:stableispop}
    Any \cstab\ matching $M$ is also \cpop\, and any super stable matching is also \spop.
\end{theorem}
\begin{proof}
   Let $M$ be a \cstab\ matching. Then, for any matching $N$, it holds that the sum of votes on each edge $e=(u,w)$ of $N$ is at least 0. Indeed, by the definition of the votes, either both are 0, or neither is. Hence, if the two votes of the endpoints are negative, then both votes are $-1$, but by the definition of the votes, such an edge is a blocking edge to $M$, a contradiction. 
   % Let $M$ be a \cstab\ matching of $G=(U,W;E)$ and consider any alternative matching $N$. Recall from the definition that, for any agent $u\in U$, if $\votec_u (M,N)=-1$, then $M(u)\ne N(u)$ and $p_u(N(u))\ge p_u(M(u))+\gamma_{N(u)}^u$. Now, notice that clearly $M(N(u))\ne N(N(u))$, so $\votec_{N(u)} (M,N)\in\{-1,+1\}$. However, if $\votec_{N(u)} (M,N)=-1$, then $p_{N(u)}(N(N(u)))\ge p_{N(u)}(M(N(u)))+\gamma_{N(N(u))}^{N(u)}$, in which case $(u,N(u))$ satisfies the definition of a $\gamma$-blocking edge, so $M$ is not \cstab\ -- a contradiction. Thus, it must be the case that $\votec_{N(u)} (M,N)=+1$. By a symmetrical argument, it follows that, for any agent $w\in W$, if $\votec_w (M,N)=-1$, then $\votec_{N(w)} (M,N)=+1$. Thus, 
  %  \begin{align*}
   %     &\{u\in U\;\vert\; \votec_u (M,N)=-1\}+\{w\in W\;\vert\; \votec_w (M,N)=-1\}\\
    %    & \leq \{u\in U\;\vert\; \votec_u (M,N)=+1\}+\{w\in W\;\vert\; \votec_w (M,N)=+1\},
    %\end{align*}
    %so $\sum_{v\in U\cup W}\votec (M,N)\geq 0$ and therefore $\Dvotec (M,N)\geq 0$ as required.
    %
  %  By a symmetrical argument, substituting $\votes$ for $\votec$ and super stable-blocking for $\gamma$-blocking, the same result follows for super stability and super popularity.
\end{proof}

Furthermore, we highlight that the following observation by \citet{irving02srt} (for the more general Stable Roommates with Ties problem) extends naturally to our problems.

\begin{prop}[\cite{irving02srt}]
\label{prop:tiebreakstable}
    Any matching $M$ is weakly stable if and only if it is stable in some instance obtained by breaking the ties. Furthermore, $M$ is super stable if and only if $M$ is stable in every instance obtained by breaking the ties.
\end{prop}

\begin{prop}
\label{prop:tiebreakpopular}
    Any matching $M$ is weakly popular if and only if it is popular in some instance obtained by breaking the ties. Furthermore, $M$ is super popular if and only if $M$ is popular in every instance obtained by breaking the ties.
\end{prop}
\begin{proof}
    It is easy to see that given a matching $M$, if we break ties in a way such that for any $e\sim_v f$ with $e\in M(v)$, we have that $e\succ'_v f$, then the vote $\votew_v(M,N)$ between $M$ and any other matching $N$ in the weakly popular setting before the tie-breaking is the same as the vote of $v$ in the strict instance obtained after tie-breaking. If there is a tie-breaking where $M$ is popular, then as $\vote_v(M,N)=+1$ implies $\votew_v(M,N)=+1$, we have that for any $N$, $0\le \Dvote (M,N)\le \Dvotew (M,N)$. Similarly, after any tie-breaking, $\vote_v(M,N)=-1$ implies $\votes_v(M,N)=-1$ in the original instance, so if there is a tie-breaking where a matching $N$ satisfies $\Dvote (M,N)<0$, then $0>\Dvote (M,N) \ge \Dvotes (M,N)$, so $M$ is not super popular. If $M$ is popular even for the tie-breaking, where for any $e\sim_v f$ with $e\in M(v)$, we have that $f\succ'_v e$ (so the vote $\votes_v(M,N)$ between $M$ and any other matching $N$ in the super popular setting before the tie-breaking is the same as the vote of $v$ in the strict instance obtained after tie-breaking), then it must be super popular.
\end{proof}

\new{It is easy to see that weakly popular matchings can be up to two times larger than weakly stable matchings, this holds even for strict preferences: just take a path with 3 edges, where the middle edge is strictly preferred by both of its vertices. Also, a super popular matching may exist, even if a super stable does not: take again a path with 3 edges, but now all edges are tied. Then, the perfect matching in the instance is super popular, but no matching is super stable.}

\section{Computing large $\gamma$-popular matchings}
\label{sec:approxalgo}

\subsection{Intractability of finding optimal solutions}

We start by highlighting the following intractability result that holds even for \wpopOpt\ and in very restricted settings, and thus also applies to \cpopOPT\ under the same restrictions.

\begin{restatable}{theorem}{NPc}
\label{thm:NPc}
    \wpopOpt\ is NP-hard, even with one-sided ties. 
\end{restatable}
\begin{proof}
    We reduce from \maxsmti, where we are given an instance $I$ of the stable matching problem with ties and a number $k$, and the question is whether there exists a stable matching with size at least $k$. In particular, we reduce from an NP-hard restricted variant, where ties occur on one side only, say on side $W$ \cite{manlove2002hard}. 
    
    Let $I=(G,\succ)$ be such an instance of \maxsmti, with $U=\{ u_1,\dots,u_n\}$ and $W=\{ w_1,\dots,w_n\}$. We create an instance $I'$ from $I$ as follows. We add an agent $z_i$ to $W$ for each $i\in [n]$ to get $W'$ and an agent $z_i'$ to $U$ for each $i\in [n]$ to get $U'$. The preference list of $z_i$ is $u_i\succ z_i'$. Agent $z_i'$ only considers $z_i$ acceptable. Finally, for each agent $u_i\in U$, we add $z_i$ to the top of their preference list. So, if $u_i$ had a preference list $w_1\succ w_2$ in $I$, now they have $z_i\succ w_1\succ w_2$. Note that in our construction, ties still occur on side $W'$ only. 

    We show that there is a stable matching of size $n-l$ in $I$ if and only if there is a \wpop\ matching of size $2n-l$ in $I'$, for any $0\le l\le n$.

\begin{claim}
    Let $M'$ be a \wpop\ matching in $I'$ of size $2n-l$. Then, there is a stable matching $M$ in $I$ of size $n-l$.
\end{claim}
\begin{claimproof}
    Clearly, $M'$ can contain at most $n$ edges from $\{ (z_i,z_i'),(u_i,z_i)\mid i\in [n]\}$, therefore $M'$ induces a matching $M$ of size at least $n-l$ in $G$. If this matching is not maximal in $G$, then let $H$ be the subgraph of the unmatched vertices of $M$ in $G$, and let $N$ be a stable matching in $H$. Then, we extend $M$ by the edges of $N$. Clearly, $M$ still has at least $n-l$ edges, and now it is maximal in $G$. 
    
    We claim that $M$ is stable. Suppose that $e=(u_i,w_j)$ blocks $M$, with $u_i\in U,w_j\in W$.

    Assume that $w_j$ is unmatched in $M'$. Then $u_i$ is matched in $M'$, which follows from $M'$ being \wpop. If $M'(u_i)\in W$, then $M'(u_i)=M(u_i)$ and $M'\setminus \{ M'(u_i)\} \cup \{ (u_i,w_j)\} $ dominates $M'$, as there is two $-1$ votes and one $+1$ vote, contradiction. If $M'(u_i)=z_i$, then in the restriction of $M'$ to $G$ both $u_i$ and $w_j$ were unmatched, so they were both in $H$. As we extended $M$ with a stable matching in $H$, we get that the edge $(u_i,w_j)$ cannot block, which is a contradiction. Hence, from now on, we may assume that $w_j$ is matched in $M'$, hence has the same partner as in $M$. Let $u_{i'}$ be the partner of $w_j$ in $M$. 
    
    If $u_i$ in unmatched in $M'$, then $M'\setminus \{ M'(w_j)\} \cup \{ (u_i,w_j)\} $ dominates $M'$, which is a contradiction. Suppose that $M'(u_i)\in W$, so $M'(u_i)=M(u_i)$. Then, we create a matching $N'$ as follows. Delete the edge $M'(u_i),M'(w_j)$ and $(z_{i'},z_{i'}')$ from $M'$, and add $(u_i,w_j)$ and $(u_{i'},z_{i'})$. Now, if we compare $N'$ to $M'$, $z_{i'},u_{i'}$ vote with $-1$, $u_i,w_j$ also vote with $-1$, as $e=(u_i,w_j)$ is a blocking edge to $M$ and hence $M'$ and only $z_{i'}'$ and $u_i$'s original partner vote with $+1$, everyone else votes with $0$. Therefore, $\Dvotew (M',N')<0$, which is a contradiction.  

    Suppose, finally, that $u_i$ is matched to $z_i$ in $M'$. Create $N'$ from $M'$ by deleting $(u_i,z_i), M'(w_j)$ and $(z_{i'},z_{i'}')$ and adding $(z_i,z_i'),(u_i,w_j)$ and $(u_{i'},z_{i'})$. Then, $z_i',w_j,u_{i'},z_{i'}$ all vote with $-1$ against $N'$ and only $z_i,u_i,z_{i'}'$ vote with $+1$, so $\Dvotew(M',N')<0$, which is a contradiction.
\end{claimproof}

\begin{claim}
    Let $M$ be a stable matching of size at least $n-l$ in $I$. Then there is a \wpop\ matching of size $2n-l$ in $I'$.
\end{claim}
\begin{claimproof}
    Let $M'$ be the matching we get from $M$ as follows. For each unmatched $u_i$ in $M$, we add the edge $(u_i,z_i)$, and for each matched $u_j$ in $M$, we add the edge $(z_j,z_j')$.

    Suppose that there is a matching $N'$ such that $\Dvotew (M',N')<0$. We may suppose that $N'$ is (inclusion-wise) maximal. 

    Observe that, since $M$ is stable, if we add the two votes on an edge $e=(u,w)\in N'$, with $u,w \in U\cup W$, then we get at least $0$. Indeed, by the definition of the votes, if one of them votes with $0$, then so does the other one, and therefore if $\votew_u (M',N') +\votew_w(M',N')<0$, then both votes are $-1$, but this means that $(u,w)$ blocks $M$, which is a contradiction. 

    Let $x$ be the number of agents in $W$ who are matched in $M'$ but not matched in $N'$, and denote them by $X$. Let $y$ be the number of agents in $U$ that are matched to a $z_i$ agent in $M'$, but are not matched to a $z_i$ agent in $N'$, and denote them by $Y$. Then, there are $x$ agents who vote with $+1$, but we do not count them on the edges of $N$. Also, there are $y$ agents from $U$, who vote with $+1$. Also, if such an $u_i$ agent has a partner $w_j\in W$ in $N'$, then $w_j$ also votes with $+1$, as otherwise the edge $(u_i,w_j)$ would have blocked $M$, as $u_i$ was unmatched in $M$, which is a contradiction. Hence, these $+1$ votes are not canceled by the other endpoints of the given edges in $N$.

    Observe that any edge of $N'$, where both endpoints vote with $-1$ must be of the form $(u_i,z_i)$. Let the number of these be $z$ and denote the union of agents in them by $Z$. Then, these edges together give us $-2z$ aggregate votes, but the original partners $z_i'$ of the $z_i$ agents vote with $+1$ in this case. Our goal is to show that $z\le x+y$. In this case, the $x$ votes of $+1$ of the agents in $X$ and the $y$ votes of $+1$ of the ones in $Y$ cancel the $z$ votes of $-1$ of the agents in $Z$, and therefore we get that $\Dvotew (M',N')\ge 0$, which is a contradiction.

    Each $u_i$ agent from the set $Z$ originally had a partner in $W$ that they are not matched to in $N'$. As the agents in $W$ can only be matched to agents from $U$, we get that because of each $u\in Z$, either there is one less agent in $W$ matched in $N'$ or one less agent in $U$ that is matched to some $z_i$. Therefore, we do indeed have $z\le x+y$. 
\end{claimproof}

Hence, it is NP-hard to decide if there is a \wpop\ matching of size at least $n+k$. 
\end{proof}

\subsection{Approximation algorithm}

In light of this strong intractability result, we will now describe how to approximate \cpopOPT\ in \nd{linear time} within a factor of $\frac{3}{4}$. Let us first provide a high-level idea of the algorithm, which is similar to recent approaches (see \cite{yokoi2021approximation,kamiyama2020popular,csaji2023approximation,csaji2023simple}).

Let $I$ be an instance of \cpopOPT.
\begin{enumerate}
    \item Create an instance $I'$ of the Stable Marriage problem with strict preferences by making parallel copies of each edge and creating strict preferences over the created edges.
    \item Run the Gale-Shapley algorithm \cite{gale1962college} to obtain a stable matching $M'$ in the new instance $I'$.
    \item Take the projection $M$ of $M'$ in $I$ by adding any edge $e\in E$ to $M$ whenever one of the parallel copies of $e$ in $I'$ is contained in $M'$.
\end{enumerate}

\medskip

We now describe the full details of the algorithm. Let $I$ be an instance of \cpopOPT\ \new{ and let $I'$ be the new instance we construct}. For each edge $e$, we create parallel copies $a(e),b(e),c(e),x(e),y(e),z(e)$ in $I'$. Then, we create strict agent preferences over the edges in $I'$ as follows. For each vertex (agent) $u\in U$, we rank the copies according to the rule 
$$ a\succeq^{\gamma} b \succ c \succ x \succeq^{\gamma} y\succ z.$$
For each vertex $w\in W$, we rank the copies according to the rule
$$ z\succeq^{\gamma} y \succ x \succ c \succeq^{\gamma} b\succ a.$$

Here, $\alpha \succ \beta$ denotes that, for any two edges $e$ and $f$, the copy $\alpha(e)$ is always ranked higher than the copy $\beta(f)$. The notion $\alpha \succeq^{\gamma} \beta $ denotes that $\alpha (e)\succ_v \beta (e)$ for any edge $e$, and $\beta (f)\succ_v  \alpha (e)$ if and only if $p_v(f)\ge p_v(e)+\gamma_f^v$. For any copy that is not $b$ or $y$, we rank the edges of the same copy according to the preference functions $p_v()$ by breaking the ties arbitrarily.

This ranking can be obtained in the following way. For each vertex $u\in U$, rank the $a,c,x,z$ copies strictly after each other according to $p_v()$ by breaking the ties arbitrarily. Then we can insert the $b$ and $y$ copies of each edge $f$ in a way such that $b(f)\succ_ua(e)$ if and only if $p_u(f)\ge p_u(e)+\gamma_f^u$, and $b(f)\succ_u c(g)$ for any edge $g$. Similarly, we can insert the $y$ copies such that $y(f)\succ_ux(e)$ if and only if $p_u(f)\ge p_u(e)+\gamma_f^u$, and $c(g)\succ_u y(f) \succ_u z(g)$ for any edge $g$. We remark that this construction allows $b$ and $y$ copies to be ranked differently than in $p_v()$, similarly to \citet{csaji2023simple}. 

\new{To illustrate this construction in a simple case, suppose that an agent $u\in U$ has the preference list (over edges) 
$$u\;:\; e\sim f\succ g$$ 
in $I$, induced by preference valuation $p_u(e)=p_u(f)=2$, $p_u(g)=1$ and let $\gamma_e^u=2$, $\gamma_f^u=\gamma_g^u=1$. Then the new preference list over the edge copies of $e,f$ and $g$ of this agent is 
\begin{align*}
  u:\;a(e)&\sim a(f)\succ b(f)\succ a(g) \succ b(e) \succ b(g)\succ c(e)\sim c(f)\succ c(g)\\
  &\succ  x(e)\sim x(f)\succ y(f)\succ x(g)\succ y(e)\succ y(g)\succ z(e)\sim z(f)\succ z(g).
\end{align*}
Breaking the ties arbitrarily, we get, for example, the following strict preference list of $u$ in $I'$:
\begin{align*}
  u:a(e)&\succ a(f)\succ b(f)\succ a(g)\succ b(e)\succ b(g)\succ c(e)\succ c(f)\succ c(g)\\
  &\succ  x(e)\succ x(f)\succ y(f)\succ x(g)\succ  y(e)\succ y(g)\succ z(e)\succ z(f)\succ z(g).
\end{align*}
}

\vspace{-1.2em}

For agents in the set $W$, we construct the preferences similarly, but using the other ranking rule we specified above. Any remaining ties in the preference system can be broken arbitrarily. 

\new{Clearly, the algorithm can be implemented in linear time, as the instance size increases by a constant factor in the original number of edges, and the Gale-Shapley algorithm runs in linear time \cite{gale1962college}.}

\begin{theorem}
\label{thm:approxguarantees}
Let $M$ be the output of this \new{algorithm. It holds that} 
\begin{itemize}
    \item $M$ is \cpop,
    \item for any matching $N$, $|M|\ge \frac{2}{3}|N|$,
    \item for any \cpop\ matching $N$, $|M|\ge \frac{3}{4}|N|$, and
    \item for any \cstab\ matching $N$, $|M|\ge \frac{4}{5}|N|$.
\end{itemize}
\end{theorem}
\begin{proof}

Let $M$ be the output of the algorithm and $M'$ be its preimage from Step 2.
 For an edge $e\in M$, we use the notation $M'(e)$ to denote the copy of $e$ that is included in $M'$.  Now, we start by showing that $M$ is indeed \cpop.

\begin{claim}
\label{claim:key}
    Let $N$ be any matching and let $e=(u,w)\in N$ be an edge. Then, the following hold:
    \begin{enumerate}
      
        \item At least one of $\{ u,w\}$ are matched \new{in $M$}.
        \item If $u$ or $w$ is unmatched in $M$, then $\votec_u (M,N)+\votec_w(M,N)=0$. Furthermore, in the first case, $M'(w)$ is of type $y$ or $z$ and in the second case, $M'(u)$ is of type $a$ or $b$.
        \item If $M'(u)$ is of type $a,b$ or $c$ and $M'(w)$ is of type $a,b$ or $c$, then $\votec_u (M,N)+\votec_w(M,N)\ge 0$. 
        \item If $M'(u)$ is of type $a,b$ or $c$ and $M'(w)$ is of type $x,y$ or $z$, then $\votec_u (M,N)+\votec_w(M,N)\ge -2$. 
        \item If $M'(u)$ is of type $x,y$ or $z$ and $M'(w)$ is of type $a,b$ or $c$, then $\votec_u (M,N)+\votec_w(M,N)=+2$. 
        \item If $M'(u)$ is of type $x,y$ or $z$ and $M'(w)$ is of type $x,y$ or $z$, then $\votec_u (M,N)+\votec_w(M,N)\ge 0$. 
        
    \end{enumerate}
\end{claim}
\begin{claimproof}
    1. If $u,w$ are both unmatched, then this means that $M$ and thus $M'$ is not maximal, contradicting that it is stable. 

    2. Suppose that $u$ is unmatched in $M$, hence also in $M'$. Then, as the edge $z(e)$ does not block, we get that $M'(w)$ is either an edge of type $y$ or of type $z$. Let $M(w)=g$. Suppose that $\votec_w(M,N)\ne +1$. Then, $p_w(e)\ge p_w(g) + \gamma_e^w$. This implies that $z(e)\succ_w z(g)\succeq_w M'(w)$, so $z(e)$ blocks $M'$, contradiction. Hence,  $\votec_w(M,N)= +1$, so $\votec_u (M,N)+\votec_w(M,N)=0$. The case when $w$ is unmatched is similar. 

    From now on, let $f=M(u)$ and $g=M(w)$.

    3. As $\votec_u (M,N)\ne 0$ and $\votec_w (M,N)\ne 0$, it suffices to show that $\votec_u (M,N)=\votec_w (M,N)=-1$ is impossible. Suppose to the contrary that this holds. Then, we know that $p_u(e)\ge p_u(f)+\gamma_e^u$ and $p_w(e)\ge p_w(g) + \gamma_e^w$. From this, we get that $b(e)\succ_u a(f)\succeq_u M'(u)$ and $b(e)\succ_w c(g)\succeq_w M'(w)$, therefore $b(e)$ blocks $M'$, which is a contradiction. 

    4. This is trivial, as both votes are at least $ -1$.

    5. Suppose that $\votec_u (M,N) =-1$. Then, $p_u(e)\ge p_u(f) + \gamma_e^u$. Hence, $y(e)\succ_u x(f)\succeq_u M'(u)$ and $y(e)\succ_w c(g)\succeq M'(w)$, therefore $y(e)$ blocks $M'$, contradiction. If $\votec_w (M,N) =-1$, then $p_w(e)\ge p_w(g) + \gamma_e^w$ and so $b(e)$ blocks $M'$, a contradiction.

    6. The proof of this last statement is analogous to the proof of point 3. 
\end{claimproof}

\begin{claim}
$M$ is \cpop.
\end{claim}
\begin{claimproof}
Suppose to the contrary that there is a matching $N$, such that $\Dvotec (M,N)<0$. Consider the symmetric difference of the two matchings, $M\triangle N$. Then, there must be a component $P$ of $M\triangle N$, such that $\Dvotec (M\cap P, N\cap P)<0$.

For each edge $f\in M\cap P$, let $t(f)=0$, if $M'(f)\in \{ a(f),b(f),c(f)\}$ and $t(f)=1$ if $M'(f)\in \{ x(f),y(f),z(f)\} $. Then, for any edge $e=(u,w)\in N$, if it is between $M$-edges $f=M(u)$ and $g=M(w)$ with $t(f)=t(g)$, by Claim \ref{claim:key} we obtain that $\votec_u (M,N) +\votec_w (M,N) \ge 0$. If $t(f)=1-t(g)=0$, then $\votec_u (M,N) +\votec_w (M,N) \ge -2$ and if $t(f)=1-t(g)=1$, then $\votec_u (M,N) +\votec_w (M,N) = +2$. 

Therefore, from $\Dvotec (M\cap P , N \cap P)<0$ and Claim \ref{claim:key} point 2, we get that there must be strictly more edges $e=(u,w)\in N\cap P$ with $t(M(u))=1-t(M(w))=0$ than with $t(M(u))=1-t(M(w))=1$. In particular, $P$ must be a path.

From the two possible orderings of the vertices (and hence also the edges) of path $P$, consider the one where the first vertex that is matched in $M$ is from $U$, and hence the last vertex matched by $M$ is from $W$. From now on, by the $k$-th edge of $M\cap P$, we mean the $k$-th $M$-edge in this ordering. 

Now, by the above, the first edge of $M$ has a type $x,y$ or $z$ copy in $M'$, and the last edge has a type $a,b$ or $c$. However, this implies that both the first and last edges of $P$ must be $M$-edges. Indeed, if the first edge were an $N$-edge, then the first vertex is an agent $w_1$, and by Claim \ref{claim:key} point 2, we obtain that the first $M$-edge must be of type $a,b$ or $c$, which is a contradiction. Similarly, if the last edge is an $N$-edge, then the last $M$-edge has to be of type $x,y$ or $z$, which is a contradiction. Hence, there are two vertices in $P$ that are only covered in $M$, and zero vertices that are only covered in $N$, so even though $\sum_{v\in V(N\cap P)}\votec_v(M,N)\ge -2$, by considering the votes of the two vertices in $V(M\cap P)\setminus V(N\cap P)$, we obtain that $\Dvotec (M\cap P,N\cap P)\ge 0$, which is a contradiction. 

Therefore, we conclude that $M$ is \cpop.
\end{claimproof}

Now we proceed to prove the approximation guarantees. 

\begin{claim}
    Let $N$ be any matching. Then, $|M|\ge \frac{2}{3}|N|$.
\end{claim}
\begin{claimproof}
\new{For any matching $N$, the symmetric difference $M\triangle N$ consists of alternating paths and cycles. Hence, it is enough to show that there is no such component $P$, where $|M\cap P|<\frac{2}{3}|N\cap P|$. Since $|M\cap P|\ge |N\cap P|-1$, this can only happen if $P$ is an alternating path with either $1$ $M$-edge and 2 $N$-edges, or a single $N$-edge.}

Clearly, $M$ is maximal, as otherwise $M'$ would not be stable, so the latter is impossible. Suppose to the contrary that there is a path $P$ in $M\triangle N$ that contains two $N$-edges $f,g$ and a middle $M$-edge $e=(u,w)$, where $f=M(u), g=M(w)$. As $M'$ was stable, we obtain that $a(f)$ does not block $M'$, and therefore $M'(e)\in \{ a(e),b(e)\}$. Similarly, as $z(g)$ does not block $M'$, we get that $M'(e)\in \{ z(e),y(e)\}$, which is a contradiction. 
\end{claimproof}

\begin{claim}
    Let $N$ be any \cpop\ matching. Then, $|M|\ge \frac{3}{4}|N|$.
\end{claim}
\begin{claimproof}
    We have already seen in Claims 2-3 that there can be no component in $N\triangle M$ with $i$ $M$-edges and $i+1$ $N$-edges (for $i\in\{0,1\}$).
    
    Suppose, for the contrary, that there is a component $P$ in $M\triangle N$ with $2$ $M$-edges and $3$ $N$-edges. Let $P=\{ w_1,u_1,w_2,u_2,w_3,u_3\}$, and let $e_i=(w_i,u_i)\in N$, for $i=1,2,3$ and $f_j=(u_j,w_{j+1})\in M$, for $j=1,2$. 
    
    As the edge $a(e_1)$ does not block $M'$ and the edge $z(e_3)$ does not block $M'$, we get that $M'(f_1)\in \{ a(f_1),b(f_1)\} $ and $M'(f_2)\in \{ z(f_2),y(f_2)\}$.

    Suppose that $M'(f_1) = a(f_1)$. Then, $b(e_2)$ blocks $M'$, which is a contradiction. Similarly, if $M'(f_2)=z(f_2)$, then $y(e_2)$ blocks $M'$, which is again a contradiction. This also means that $p_{u_1}(f_1)\ge p_{u_1}(e_1)+\gamma_{f_1}^{u_1}$ and $p_{w_3}(f_2)\ge p_{w_3}(e_3)+\gamma_{f_2}^{w_3}$ (using that $a(e_1)$ and $z(e_3)$ do not block $M'$). 

    We conclude that $M'(f_1)= b(f_1), M'(f_2)=y(f_2)$. Now, because $x(e_2)$ does not block $M'$, we get that $p_{u_2}(f_2)\ge p_{u_2}(e_2) + \gamma_{f_2}^{u_2}. $ The same argument for $c(e_2)$ gives $p_{w_2}(f_1)\ge p_{w_2}(e_2)+\gamma_{f_1}^{w_2}$. 

    Putting this together, we find that if we compare the matching $N$ to $\hat{N}=N\setminus \{ (w_i,u_i) \mid i=1,2,3\} \cup \{ (u_i,w_{i+1})\mid i=1,2\}$, then agents $u_1,w_2,u_2$ and $w_3$ vote with $-1$, as their improvement in $\hat{N}$ is large enough from $N$ and only $w_1$ and $u_3$ vote with $+1$, all other agents vote with $0$. Hence, $\Dvotec (N,\hat{N})<0$, contradicting that $N$ is \cpop. 
\end{claimproof}

\begin{claim}
    For any \cstab\ matching $N$, we have that $|M|\ge \frac{4}{5}|N|$. 
\end{claim}
\begin{claimproof}
    Let $N$ be a \cstab\ matching. As any \cstab\ matching is also \cpop, by the previous proofs, there can be no component in $N\triangle M$ with $i$ $M$-edge and $i+1$ $N$-edge (for $i\in\{0,1,2\}$).

    Suppose, for the contrary, that there exists a component $P$ in $M\triangle N$ with $3$ $M$-edges and $4$ $N$-edges. Let $P=\{ w_1,u_1,w_2,u_2,w_3,u_3,w_4,u_4\}$, where $e_i=(w_i,u_i)\in N$ for $i=1,2,3,4$ and $f_j=(u_ijw_{j+1})\in M$, for $j=1,2,3$. 

    Similar arguments as before give that $M'(f_1)\in \{ a(f_1),b(f_1)\}$ and $M'(f_3)\in \{ z(f_3),y(f_3)\}$. Now, it is clear that $M'(f_2)$ is either an $a,b$ or $c$ copy, or an $x,x$ or $z$ copy. Suppose that $M'(f_2) \in \{ a(f_2),b(f_2),c(f_2)\}$. Then, $w_3$ prefers $x(e_3)$ to $M'(f_2)$. As $x(e_3)$ does not block $M'$, we get that it must be dominated at $u_3$, so $M'(f_3)=y(f_3)$ and $p_{u_3}(f_3)\ge p_{u_3}(e_3)+\gamma_{f_3}^{u_3}$. Now knowing that $M'(f_3)=y(f_3)$, then we get that $p_{w_4}(f_3)\ge p_{w_4}(e_4)+\gamma_{f_3}^{w_4}$, as $z(e_4)$ does not block. Combining this, we get that $f_3$ \cgblocks\ $N$, which is a contradiction. 

    The case for $M'(f_2) \in \{ x(f_2),y(f_2),z(f_2)\}$ is symmetric and leads to $f_1$ being a \cgblocking\ edge to $N$, which is a contradiction. 
\end{claimproof}

The theorem follows from the above claims. 
\end{proof}

We now show that our approximation analysis is tight in all three cases, even in the case of ordinary weak preferences and one-sided ties.

\begin{example}
    Consider an instance with agents $\{ w_1,u_1,w_2,u_2,w_3,u_3\}$. Let the edges be $E\cup F$, where $E=\{ e_i=(u_i,w_i)\mid i=1,2,3\}$, $F=\{ f_i=(u_i,w_{i+1})\mid i=1,2\}$. \nd{Define the preferences such that all agents with degree $2$ strictly prefer the $f$ edge to $e$ edge. In this instance, all preferences are strict, so weak popularity and popularity coincide. It is easy to see that the only popular matching is $M=F$,} so this must be the output of the algorithm. However, for $N=E$, $\frac{2}{3}|N|=|M|$.
\end{example}

\begin{example}
    Consider an instance with agents $\{ w_1,u_1,w_2,u_2,w_3,u_3,w_4,u_4\}$. Let the edges be $E\cup F$, where $E=\{ e_i=(u_i,w_i)\mid i=1,2,3,4\}$, $F=\{ f_i=(u_i,w_{i+1})\mid i=1,2,3\}$. Let the preference lists be given by 
    \new{\[
    \begin{minipage}{0.45\textwidth}
    \begin{align*}
    u_1&: f_1\succ e_1 \\
    u_2&: f_2\succ e_2 \\
    u_3&: f_3\succ e_3 \\
    u_4&: e_4
    \end{align*}
    \end{minipage}
    \hfill
    \begin{minipage}{0.45\textwidth}
    \begin{align*}
    w_1&: e_1 \\
    w_2&: f_1\sim e_2 \\
    w_3&: f_2\sim e_3 \\
    w_4&: f_3\succ e_4
    \end{align*}
    \end{minipage}
    \]
    where $\succ$ indicates a strict preference and $\sim$ indicates a tie.} First, we show that $N=E$ is weakly popular. It is easy to see that only agents $u_1,u_2,u_3$ and $w_4$ may vote $-1$ when comparing $N$ to any matching $M$. However, whenever $u_1$ improves, $w_1$ is worse off and votes with $+1$. Similarly, whenever $u_2$ or $u_3$ improves, $w_2$ and $w_3$ are weakly worse off, respectively, so they vote with $+1$. Finally, if $w_4$ improves, then $u_4$ is worse off and votes with $+1$. Hence, the number of $+1$ votes is always at least the number of $-1$ votes, showing that $N$ is weakly popular. 

    Next, we show that the algorithm may return $M=F$. Suppose that $M'(f_1)=b(f_1),M'(f_2)=c(f_2)$ and $M'(f_3)=y(f_3)$. Furthermore, suppose that the ties are broken up in the following way: $w_2:f_1\succ e_2$, $w_3:f_2\succ e_3$, which implies that $b(f_1)\succ_{w_2}b(e_2)$ and $c(f_2)\succ_{w_3}c(e_3)$.
    It is straightforward to verify that $M'$ is stable, and $M'$ gives $M$ when we project it back. As $|M|=\frac{3}{4}|N|$, we get that this approximation ratio is tight for the algorithm. 
\end{example}

\begin{example}
    Consider an instance with agents $\{ w_1,u_1,w_2,u_2,w_3,u_3,w_4,u_4,w_5,u_5\}$. Let the edges be $E\cup F$, where $E=\{ e_i=(u_i,w_i)\mid i=1,2,3,4,5\}$, $F=\{ f_i=(u_i,w_{i+1})\mid i=1,2,3,4\}$. Let the preference lists be given by
    \new{\[
    \begin{minipage}{0.45\textwidth}
    \begin{align*}
    u_1&: f_1\succ e_1 \\
    u_2&: f_2\succ e_2 \\
    u_3&: f_3\succ e_3 \\
    u_4&: e_4\succ f_4 \\
    u_5&: e_5
    \end{align*}
    \end{minipage}
    \hfill
    \begin{minipage}{0.45\textwidth}
    \begin{align*}
    w_1&: e_1 \\
    w_2&: f_1\sim e_2 \\
    w_3&: f_2\sim e_3 \\
    w_4&: f_3\sim e_4 \\
    w_5&: f_4\succ e_5
    \end{align*}
    \end{minipage}
    \]}
    
    We claim that $N=E$ is stable. Indeed, no edge in $F$ blocks $N$. Now let $M=F$. We claim that the algorithm may return $F$. Let $M'(f_1)=b(f_1), M'(f_2)=c(f_2), M'(f_3)=x(f_3)$ and $M'(f_4)=y(f_4)$. Suppose that we break the ties in the following manner: $w_2: f_1\succ e_2$, $w_3: f_2\succ e_3$, $w_4:f_3\succ e_4$. Then, it is straightforward to verify that $M'$ is stable and that $M'$ gives $M$. Also, as $|M|=\frac{4}{5}|N|$, our approximation analysis is tight in this case, too. 
\end{example}

\subsection{Inapproximability result}

Finally, we highlight that, assuming either the strong-Unique Games Conjecture (strong-UGC) or the Small Set Expansion Hypothesis (SSEH), it is NP-hard to approximate \wpopOpt\ within a factor of $\frac{3}{4}+\varepsilon$ (for any $\varepsilon >0$). 

\begin{restatable}{theorem}{wpopinappr}
\label{thm:wpopinappr}
    Assuming either the strong-UGC or the SSEH, it is NP-hard to approximate \wpopOpt\ within a factor of $(\frac{3}{4}+\varepsilon)$, for any $\varepsilon >0$.
\end{restatable}

We utilise that \cite{smti1.5inapprox} have shown the following result. 

\begin{theorem}[\cite{smti1.5inapprox}]
\label{thm:EMM-inappr}
    If the strong-UGC or the SSEH holds, then for any $\varepsilon >0$, it is NP-hard to distinguish for a bipartite graph $G=(U,W,E)$ with $|U|=|W|=n$ the following two cases: \begin{itemize}
        \item There is a maximal matching in $G$ of size at most $(\frac{1}{2}+\varepsilon )n$ or
        \item every maximal matching in $G$ has at least $(1-\varepsilon )n$ edges.
    \end{itemize}
\end{theorem}

Now, we are ready to prove Theorem~\ref{thm:wpopinappr}. 

\begin{proof}[Proof of Theorem~\ref{thm:wpopinappr}]
    We show that for any $\varepsilon >0$, such an approximation algorithm for \wpopOpt\ implies that we can distinguish between the two cases of Theorem \ref{thm:EMM-inappr}, which is a contradiction. 

    Let $\varepsilon >0$ and let $G=(U,W,E)$ be a bipartite graph with $|U|=|W|=n$, where $U=\{ u_1,\dots,u_n\}$ and $W=\{ w_1,\dots,w_n\}$, call this instance $I$. We may assume that $n$ is even, otherwise we duplicate $G$ and get an equivalent problem. We create an instance $I'$ of \wpopOpt. We create sets of agents $W'=\{ w_1',\dots,w_n'\}$, $U'=\{ u_1',\dots,u_{\frac{n}{2}}'\}$, $Z= \{ z_1,\dots,z_{\frac{n}{2}}\}$ and $Z'= \{ z_1',\dots,z_{\frac{n}{2}}'\}$. We create the edges of the instance in $I'$, by keeping the edges induced by $G$, and adding $\{ (w_i,w_i')\mid i\in [n]\} \cup \{ (u_j',z_j),(z_j,z_j')\mid j\in [\frac{n}{2} ] \} \cup \{ (u_i,u_j')\mid i\in [n],j\in [\frac{n}{2}]\}$.

    The preferences are illustrated in Figure \ref{tab:inappproxprefs}. Because the created graph is simple, we write the weak preferences over the set of neighbors $V(u)$ for an agent $u$ instead of over $E(u)$. $V_G(v)$ denotes the set of neighbors of $v$ in $G$, and $[X]$ for a set $X$ means that the elements of $X$ form a single tie in the preference list.

    \begin{figure}
        \centering
        \begin{tabular}{cl||cl}
            $u_i:$ & $[V_G(u_i)]\succ [U']$ & $w_i:$ & $[V_G(w_i)] \succ w_i'$  \\
             $w_i':$ & $w_i$ & $u_j':$  & $z_j\succ [U]$ \\
            $z_j:$ & $u_j'\succ z_j'$ & $z_j':$ & $z_j$ \\
        \end{tabular}
        \caption{The created preferences in Theorem \ref{thm:wpopinappr}. Here, $i\in [n],j\in [\frac{n}{2}]$}
        \label{tab:inappproxprefs}
        \Description{A two-sided preference construction.}
    \end{figure}

Let $k\ge \frac{n}{2}$. 

\begin{claim}
\label{otherdir}
    Suppose that there is a \wpop\ matching $M'$ in $I'$ of size at least $\frac{5}{2}n-k$. Then, there is a maximal matching in $G$ of size at most $k$.
\end{claim}
\begin{claimproof}
    Let $M$ be the restriction of $M'$ to $G$. Suppose that $|M|=k'>k$. Then, we get that the size of $M'$ is at most $(n-k')+k' + (n-k')+\frac{n}{2}=\frac{5}{2}n-k'<\frac{5}{2}n-k$, which is a contradiction. 

    Therefore, it is enough to show that $M$ is maximal. Suppose that $(u_i,w_j)$ is an edge of $G$ that can be added to $M$. If $u_i$ or $w_j$ is unmatched in $M'$, then at most one of them leaves a partner alone if they get matched together and strictly improve, so we obtain a matching that dominates $M'$, which is a contradiction. 

    Hence, $(w_j,w_j')\in M'$ and $(u_i,u_l')\in M'$ for some $l\in [\frac{n}{2}]$. Also, $(z_l,z_l')\in M'$ since a \wpop\ matching must be maximal. Let $N'$ be the matching $M'\setminus \{ (w_j,w_j'),(u_i,u_l'),(z_l,z_l')\cup \{ (u_i,w_j),(u_l',z_l)\}$. Then, $u_i,w_j,u_l',z_l$ all improve and vote with $-1$ and only $w_j'$ and $z_l'$ vote with $+1$. Hence, $N'$ dominates $M'$, which is a contradiction. 
\end{claimproof}

\begin{claim}
\label{onedir}
    If there is a maximal matching of size $k$ in $G$, then there is a \wpop\ matching of size at least $\frac{5}{2}n-k$ in $I'$.
\end{claim}
\begin{claimproof}
    Let $M$ be such a matching of $G$. Extend it to a matching $M'$ of $I'$ as follows. For each $i\in [n]$, such that $w_i$ is unmatched in $M$, add $(w_i,w_i')$. Let $\{ u_{i_1},\dots ,u_{i_{n-k}}\}$ be the unmatched agents in $U$ by $M$. Then, we add the edges $(u_{i_j},u_j'),(z_j,z_j')$ for $j\in [n-k]$. As $n-k\le \frac{n}{2}$, we can do this. Finally, if $n-k<\frac{n}{2}$, then for each $n-k<j\le \frac{n}{2}$ we add $(u_j',z_j)$. Clearly $|M'|=(n-k)+k+2(n-k)+(\frac{n}{2}-(n-k))=\frac{5}{2}n-k$. 

    We show that $M'$ is \wpop. From the fact that $M$ is inclusion-wise maximal, we get that the only blocking edges to $M'$ are the edges $(u_j',z_j)$, $j\in [n-k]$. Therefore, these are the only possible $(-1,-1)$ voting edges in another matching. 

    Suppose that a matching $N'$ dominates $M'$. Consider the components of the symmetric difference $N'\triangle M'$. Then, there is a component such that within this component, $N'$ also dominates $M'$. In such a component, there must be an edge $e$ of $N'$, such that both of its endpoints vote with $-1$. As we have seen, such an edge must be $(u_j',z_j)$, for some $j\in [n-k]$. As for a vertex $z_j$, its only other neighbor is $z_j'$, which is a leaf vertex, we get that at most two such edges can be in any component, and the component cannot be a cycle. Also, if $(u_j',z_j)\in N'$, then $z_j'$ is unmatched in $N'$, but is matched in $M'$, so such edges can only be in components in $M'\triangle N'$, where they are adjacent to an $M'$-edge, whose other endpoint is only covered by $M'$. 

    Let $(u_j',z_j)\in N'$. If $u_{i_j}$ is unmatched in $N'$, then they and $z_j'$ vote with $+1$ in the component and only $u_j'$ and $z_j$ vote with $-1$, which is a contradiction.
    
    If $(u_{i_j},u_l')\in N'$, then both of them vote with $+1$, so there is a $(+1,+1)$ edge. Therefore, even if there are two $(-1,-1)$ edges of $N'$ in the component, the $(+1,+1)$ edge and the two $+1$ votes by the endpoints only covered by $M'$ cancel those, which is a contradiction. 

    So $(u_{i_j},w_i)\in N'$ for some $i\in [n]$. As $M$ was maximal and $(u_{i_j},u_j')\in M'$, we get that $w_i$ was matched in $M$, so $(u_l,w_i)\in M'$ for some $l\in [n]$. This implies that $u_l$ and $w_i$ both vote with $+1$ in any case. Now, if $u_l$ is unmatched in $N'$, then there is only one $(-1,-1)$ edge in the component, but both endpoints of the component vote with $+1$, which is a contradiction. If $u_l$ and their partner in $N'$ both vote with $+1$, then we have a $(+1,+1)$ edge, and we get a contradiction by a similar reasoning as before. Otherwise, the partner of $u_l$ in $N'$ must strictly improve from $M'$, hence $(u_l,w_m)\in N'$ for an agent $w_m$ who was unmatched in $M$, so $(w_m,w_m')\in M'$. As $w_m'$ is a leaf, we get that in this case, this is the whole component. However, as both endpoints are only covered by $M'$, we get that those two $+1$ votes cancel the $(-1,-1)$ edge, which is a contradiction. 

    Hence, we conclude that $M'$ is \wpop.
\end{claimproof}
\medskip

Now suppose for contradiction that there is a $(\frac{3}{4}+\varepsilon)$-approximation algorithm for \wpopOpt. We may assume that $n$ is larger than $\frac{1}{\varepsilon}$. Now, if every maximal matching of $G$ has size at least $(1-\varepsilon )n$, then any \wpop\ matching in $I'$ has size at most $\frac{5}{2}n-(1-\varepsilon )n$ by Claim \ref{otherdir}. Also, if there is a maximal matching of size at most $(\frac{1}{2}+\varepsilon)n>\frac{n}{2}+1$ edges, then there is one with $\frac{n}{2}\le k<(\frac{1}{2}+\varepsilon)n$ edges (if $k<\frac{n}{2}$, then we increase it by one via an alternating path and get a larger maximal matching, as we may assume that $G$ has a matching of size at least $\frac{n}{2}$, otherwise distinguishing the cases is easy), and by Claim \ref{onedir} we have that there is a weakly popular matching of size at least $\frac{5}{2}n-k\ge (2-\varepsilon )n$. 

Now, because $(2-\varepsilon )n \cdot (\frac{3}{4}+\varepsilon )= \frac{3}{2}n-\frac{3}{4}\varepsilon n+ 2\varepsilon n - \varepsilon^2 n>(\frac{3}{2}+\varepsilon )n=\frac{5}{2}n-(1-\varepsilon )n$ for $\varepsilon <\frac{1}{2}$, we get that by Claim \ref{onedir}, if the first case of Theorem \ref{thm:EMM-inappr} holds, then we can find a maximal matching of size strictly less than $(1-\varepsilon )n$, which is a contradiction. 
\end{proof}

\section{Verification and exact computation}
\label{sec:verify}

\subsection{\nd{Verifying $\gamma$-popularity}}

\nd{We now turn our attention to the verification of popularity. We start with a positive result regarding the verification of $\gamma$-popularity, which also applies to weak popularity (as $\gamma$-popularity is strictly more general).}

\begin{theorem}
    We can decide if a matching is \cpop\ \nd{in an instance with $n$ agents in $O(n^3)$} time.
\end{theorem}
\begin{proof}
    Consider an instance consisting of a graph $G=(U,W;E)$, valuations $(p_v())_{v\in U\cup W}$, and numbers $\gamma_e^v$ for every $(e,v)\in E\times (U\cup W)$ with $v\in e$.

    Let $M$ be a matching of the instance and suppose without loss of generality that $|U|\ge |W|$. We create a complete bipartite graph $G_M=(U,W';E')$ from $G$ and $M$, where $W'$ is obtained from $W$ by adding $|U|-|W|$ dummy vertices \new{and $E'=U\times W'$}, and a weight function $\omega$ on the edges $E'$ of $G_M$ as follows. Clearly, $M$ is also a matching in $G_M$. For each $e=(u,w)\in E'$, if also $e\in E$, then let $\omega (u,w)=\votec_u (M(u),e)+\votec_w (M(w),e)$. Otherwise, let $\omega (u,w) = \votec_u( M(u),\emptyset) +\votec_w(M(w),\emptyset)$.
    
    Then, for any matching $N$ in $G$, we can extend it to a perfect matching $N'$ in $G_M$, such that \new{
    \begin{align*}
        \Dvotec (M,N) &=\sum\limits_{v\in U\cup W}\votec_v(M(e),N(e))=\sum\limits_{e\in N}\omega (e) + \sum\limits_{v:N(v)=\emptyset}\votec_v(M(v),\emptyset)\\
        &\ge \omega(N) + \sum\limits_{e\in N'\setminus N}\omega(e)=\sum\limits_{e\in N'}\omega (e)=\omega (N'),
    \end{align*}}
    using that, for any $e\in E$, $\votec_v(M(v),\emptyset)\ge \votec_v(M(v),e)$. Furthermore, if $N'\setminus N$ only contains edges from $E'\setminus E$, then we have $\Dvotec(M,N)=\omega (N')$.

    Therefore, on one hand, if there is a matching $N$ in $G$ such that $\Dvotec(M,N)<0$, then there exists a perfect matching $N'$ in $G_M$ with negative $\omega$-weight. On the other hand, if there exists a perfect matching $N'$ in $G_M$ with $\omega (N')<0$, then taking $N=N'\cap E$, we have $\Dvotec(M,N)=\omega (N')$ by the above, so $N$ dominates $M$. 

    Hence, deciding if $M$ is \cpop\ reduces to a minimum weight perfect matching problem, which can be solved \nd{in at most cubic time in the number of vertices for bipartite graphs} with the well-known Hungarian method \cite{hungarian}.
\end{proof}

\subsection{Exact algorithm \new{for \cpopOPT}}
\label{sec:exactalgo}

\nd{Now let us turn to the computation of optimal maximum-size $\gamma$-stable and $\gamma$-popular matchings.} We use the notation $e\triangleleft_v f$ to denote that $p_v(e)\ge p_v(f)+\gamma_e^v$. Now consider the following ILP. 

\[
\begin{aligned}
\text{maximize}\quad & \sum_{e\in E} x_e \\
\text{subject to}\quad
& \sum_{e\in E(v)} x_e \le 1
    \qquad\forall v\in U\cup W
    &&\\[6pt]
& \sum_{\substack{f\in E(u)\\ e \,\ntriangleleft_u, f}} x_f
  + \sum_{\substack{f\in E(w)\\ e \,\ntriangleleft_w, f}} x_f
  + x_e \ge 1
    \qquad\forall e=(u,w)\in E
    &&\text{(no blocking edges)}\\[6pt]
& x_e\in\{0,1\}\qquad\forall e\in E.
\end{aligned}
\]

Clearly, any solution to this ILP corresponds to the characteristic vector of a matching. Furthermore, it is easy to see that $\sum_{\substack{f\in E(u)\\ e \,\ntriangleleft\, f}} x_f
  + \sum_{\substack{f\in E(w)\\ e \,\ntriangleright\, f}} x_f
  + x_e \ge 1$ 
is satisfied for a characteristic vector $x$ of a matching $M$ if and only if $e$ is not $\gamma$-min blocking. Indeed, if $e=(u,w)$ $\gamma$-min blocks $M$, then $x_e=0$, $u$ has no edge or an edge $f$ with $e\triangleleft_u f$, and, similarly, $w$ either has no edge in $M$ or has an edge $g$ with $e\triangleleft_w g$. Either way, the constraint is violated as each term \new{on the left side of the inequality} is 0. \new{Conversely}, if the constraint is violated for some $e=(u,w)$, then all terms \new{on the left side of the inequality} are $0$, so $e\notin M$, $e\triangleleft_uM(u)$, and $e\triangleleft_w M(w)$, so $e$ must $\gamma$-min block $M$. Hence, we proved the following.
\begin{lemma}
    A matching $M$ is $\gamma$-min stable if and only if its characteristic vector is a solution to the above ILP. Therefore, the optimal solutions correspond to the maximum-size $\gamma$-min stable matchings. \qed
\end{lemma}

Now consider the following algorithm for finding a maximum-size $\gamma$-popular matching. Let $I$ be an instance of \cpopOPT. First, we create an instance $I'$ by duplicating each edge with copies $a(e)$ and $x(e)$. For each vertex $u\in U$, we set $p_u(a(e))= p_u(e)+R$ and $p_u(x(e))=p_u(e)$, where $R=\max \{ p_v(f)\mid v\in U \cup W, f\in E(v)\}$ is large enough to ensure that any $a$ type copy is better than any $x$ type copy. \new{For each vertex $w\in W$, we set} $p_w(a(e))= p_w(e)$ and $p_w(x(e))=p_w(e)+R$. Furthermore, we let $\gamma_{a(e)}^v=\gamma_{x(e)}^v=\gamma_e^v$ for $(e,v)\in E\times (U\cup W)$ with $v\in e$.
%In contrast to our previous algorithms, now we do not break ties, and let the ranking within the $a$ copies and within the $b$ copies both remain the same as the original ranking. For example, if we had $e\sim_uf\succ_ug$, then we get $a(e)\sim_ua(f)\succ_ua(g)\succ_ub(e)\sim_ub(f)\succ_ub(g)$. 
Finally, we find a maximum-size $\gamma$-min stable matching $M'$ in $I'$ \new{(e.g., using the ILP we described above)} and output its projection $M$ to $I$.

The fact that the output of the algorithm is a maximum-size \cpop\ matching follows from the theorem below.

\begin{restatable}{theorem}{exactopt}
\label{thm:exactopt}
    The projection $M$ of a $\gamma$-min stable matching $M'$ in $I'$ is $\gamma$-popular. Furthermore, there exists a maximum-size \cpop\ matching $M$ that is a projection of a $\gamma$-min stable matching $M'$ in $I'$.
\end{restatable}
\begin{proof}
Let $M'$ be a $\gamma$-min stable matching and $M$ be its projection.
 
 \begin{claim}
\label{claim:key-new}
    Let $e=(u,w)\in E$ be an edge. Then, the following hold:
    \begin{enumerate}
      
        \item At least one of $\{ u,w\}$ are matched.
        \item If $u$ or $w$ is unmatched in $M'$, then $\votec_u (M,N)+\votec_w(M,N)=0$. Furthermore, in the first case, $M'(w)$ is of type $x$ and in the second case, $M'(u)$ is of type $a$.
        \item If $M'(u)$  and $M'(w)$ are of type $a$, then $\votec_u (M,N)+\votec_w(M,N)\ge 0$. 
        \item If $M'(u)$ is of type $a$ and $M'(w)$ is of type $x$ too, then $\votec_u (M,N)+\votec_w(M,N)\ge -2$. 
        \item If $M'(u)$ is of type $x$ and $M'(w)$ is of type $a$, then $\votec_u (M,N)+\votec_w(M,N)=+2$. 
        \item If $M'(u)$ and $M'(w)$ are of type $x$, then $\votec_u (M,N)+\votec_w(M,N)\ge 0$. 
        
    \end{enumerate}
\end{claim}
\begin{claimproof}

    1. If $u,w$ are both unmatched, then this means that $M'$ is not maximal, contradicting that it is $\gamma$-min stable. 

    2. Suppose that $u$ is unmatched in $M'$. Then, as the edge $x(e)$ does not block, we get that $M'(w)$  of type $x$. Let $M(w)=g$. Suppose that $\votec_w(M,N)\ne +1$. Then, $p_w(e)\ge p_w(g) + \gamma_e^w$, and hence $p_w(x(e))=p_w(e)+R\ge p_w(M'(w))+\gamma_e^w$, so $x(e)$ $\gamma$-min blocks $M'$, contradiction. Hence,  $\votec_w(M,N)= +1$, so $\votec_u (M,N)+\votec_w(M,N)=0$. The case when $w$ is unmatched is similar. 

    From now on, let $f=M(u)$ and $g=M(w)$.

    3. As $\votec_u (M,N)\ne 0$ and $\votec_w (M,N)\ne 0$, it is enough to show that $\votec_u (M,N)=\votec_w (M,N)=-1$ is impossible. Suppose, for the contrary, that this holds. Then, we know that $p_u(e)\ge p_u(f)+\gamma_e^u$ and $p_w(e)\ge p_w(g) + \gamma_e^w$, therefore  $p_u(a(e))\ge p_u(a(f))+\gamma_e^u$ and $p_w(a(e))\ge p_w(a(g)) + \gamma_e^w$ and $a(e)$ $\gamma$-min blocks $M'$, contradiction. 

    4. This is trivial, as both votes are at least $ -1$.

    5. Suppose that $\votec_u (M,N) =-1$. Then, $p_u(e)\ge p_u(f) + \gamma_e^u$, so $p_u(x(e))=p_u(e)\ge p_u(x(f))+\gamma_e^u$, therefore $x(e)$ $\gamma$-min blocks $M'$, contradiction. If $\votec_w (M,N) =-1$, then $p_w(a(e))\ge p_w(a(g)) + \gamma_e^w$ and so $a(e)$ $\gamma$-min blocks $M'$, a contradiction.

    6. The proof of the last statement is analogous to the proof of point 3.    
\end{claimproof}

We will now establish $\gamma$-popularity. Suppose that there exists a matching $N$ such that $\Dvotec (M,N)<0$. Then, there must exist a component $P$ of the symmetric difference $M\Delta M$, such that $\Dvotec (M\cap P, N\cap P)<0$. 

For each edge $f\in M\cap P$, let $t(f)=0$, if $M'(f)=a(f)$ and $t(f)=1$ if $M'(f)=x(f) $. Then, for any edge $e=(u,w)\in N$, if it is between $M$-edges $f=M(u)$ and $g=M(w)$ with $t(f)=t(g)$, by Claim \ref{claim:key-new} we obtain that $\votec_u (M,N) +\votec_w (M,N) \ge 0$. If $t(f)=1-t(g)=0$, then $\votec_u (M,N) +\votec_w (M,N) \ge -2$ and if $t(f)=1-t(g)=1$, then $\votec_u (M,N) +\votec_w (M,N) = +2$. Therefore, from $\Dvotec (M\cap P , N \cap P)<0$ and Claim \ref{claim:key-new} point 2, we get that there must be strictly more edges $e=(u,w)\in N\cap P$ with $t(M(u))=1-t(M(w))=0$ than with $t(M(u))=1-t(M(w))=1$. In particular, $P$ must be a path.

From the two possible orderings of the vertices (and hence also the edges) of $P$, consider the one where the first vertex that is matched in $M$ is from $U$, and hence the last vertex matched by $M$ is from $W$. From now on, by the $k$-th edge of $M\cap P$, we mean the $k$-th $M$-edge in this ordering. Now, by the above, the first edge of $M$ has a type $a$ copy in $M'$, and the last edge has a type $x$. However, this implies that both the first and last edges of $P$ must be $M$-edges.

Indeed, if the first edge were an $N$-edge, then the first vertex is an agent $w_1$, and by Claim \ref{claim:key-new} point 2, we obtain that the first $M$-edge must be of type $a$, which is a contradiction.

Similarly, if the last edge is an $N$-edge, then the last $M$-edge has to be of type $x$, which is a contradiction. Hence, there are two vertices in $P$ that are only covered in $M$ and zero vertices that are only covered in $N$, so even though $\sum_{v\in V(N\cap P)}\votec_v(M,N)\ge -2$, by considering the votes of the two vertices in $V(M\cap P)\setminus V(N\cap P)$, we obtain that $\Dvotec (M\cap P,N\cap P)\ge 0$, which is a contradiction. Therefore, we conclude that $M$ is \cpop.

 %We create instances $I_{\mathrm{strict}}$ and $I'_{\mathrm{strict}}$ with strict preferences as follows. For each $e=(u,w)$ in $M$, we break the ties in $I_{\mathrm{strict}}$ for $u$ and $w$ such that the edge $e$ is preferred to any edge it was tied with in $I$. Then, for $I'_{\mathrm{strict}}$ we use this tie-breaking for both within the $a$ and within the $b$ copies. 
 
Finally, we show that there exists a maximum-size \cpop\ matching that is a projection of a $\gamma$-min stable matching. Let $M$ be a maximum-size \cpop\ matching of $I$ and create an instance $I_{\mathrm{strict}}$ with strict preferences and the same set of edges as $I$ as follows. For each $v\in U\cup W$, we create an arbitrary preference list $\succ_v$ that satisfies that, for $e=M(v)$, if $p_v(f)<p_v(e)+\gamma_f^v$ then $e\succ_v f$, and if $p_v(f)\ge p_v(e)+\gamma_e^v$ then $f\succ_v e$. Since, for any edge $f$, exactly one of these holds, this is possible. 

We first claim that $M$ is popular in $I_{\mathrm{strict}}$. This easily follows from the fact that $\votec_v(M,f)=\vote_v^{I_{\mathrm{strict}}}(M,f)$.

As shown in \cite{cseh2018dominant}, Section 3, there exists a \emph{dominant} matching $M_D$ (a matching $M_D$ is dominant, if it is popular, and for any matching $N$ with $|N|>|M_D|$, $\Dvote (M,N)>0$) in $I_{\mathrm{strict}}$. By definition, dominant matchings are always maximum-size popular matchings, so $|M_D|\ge |M|$. From  $\votec_v(M,f)=\vote_v^{I_{\mathrm{strict}}}(M,f)$, we also know that $M_D$ is a \cpop\ matching in $I$.

Again, as shown in \cite{cseh2018dominant} Section 3, $M_D$ is a projection of a stable matching $M_D'$ of an instance $I'_{\mathrm{strict}}$, which is created from $I_{\mathrm{strict}}$ by (i) duplicating every edge with copies $a$ and $b$, (ii) ranking the edges in the same order as in $I_{\mathrm{strict}}$ within each type, and finally (iii) for vertices in $U$ any type $a$ edge is preferred to any type $x$ edge and for vertices in $W$, any type $b$ edge is preferred to any type $a$ edge. This instance is very similar to the instance $I'$, where only the preferences within type $a$ and type $x$ copies differ, but the sets of vertices and edges are the same. Hence, $M_D'$ gives a matching in $I'$ too.

We claim that $M_D'$ is $\gamma$-min stable in $I'$. Suppose, \new{for the sake of contradiction,} that an edge $a(e)=a(u,w)$ $\gamma$-min blocks $M_D'$ in $I'$. This implies that $w$ must be matched with a type $a$ copy $a(g)$ (or $w$ remains unmatched), and $p_w(a(e))\ge p_w(a(g))+\gamma_{a(e)}^w$. Therefore, if $w$ is matched, then $p_w(e)\ge p_w(g)+\gamma_e^w$, implying that $e\succ_w g $ in $I_{\mathrm{strict}}$ and hence $a(e)\succ_w a(g)$ in $I'_{\mathrm{strict}}$. Furthermore, if $u$ is matched with an $a$ type copy $a(f)$, then $p_u(a(e))\ge p_u(a(f))+\gamma_{a(e)}^u$. Therefore, $p_u(e)\ge p_w(g)+\gamma_e^w$, implying that $e\succ_u f $ in $I_{\mathrm{strict}}$, and hence $a(e)\succ_u a(f)$ in $I'_{\mathrm{strict}}$. If $u$ is not matched with an $a$ type copy, then $a(e)$ is clearly preferred to $M_D'(u)$ in $I'_{\mathrm{strict}}$. Hence, we can observe that, in any case, we get that $a(e)$ blocks $M_D'$ in $I'_{\mathrm{strict}}$, which is a contradiction. A similar argument leads to a contradiction for $x$-type edges too.

Therefore, we conclude that $M_D$ is a maximum-size \cpop\ matching in $I$ that is a projection of a $\gamma$-min stable matching $M_D'$, as required.
\end{proof}

\section{Super popular matchings}
\label{sec:superpop}

Unlike the computation of weakly popular and $\gamma$-popular matchings (of arbitrary sizes), we note that the computation of super popular matchings is computationally intractable. In fact, even deciding whether a \spop\ matching exists is NP-hard (as super popular matchings might not exist). \nd{This also contrasts the complexity of the problem of deciding the existence of super stable matchings, which can be done in polynomial time \cite{irving1994stable}. As we showed in Theorem \ref{thm:stableispop} that super stable matchings are super popular, the intractable case is really whenever instances do not admit super stable matching.} 

\begin{restatable}{theorem}{superpop}
\label{thm:superpop}
    Deciding if an instance $I$ of \spopex\ admits a super popular matching is NP-hard even with ordinary weak preferences, where only two agents have ties, and both have just a single tie of length two. 
\end{restatable}
\begin{proof}
    We reduce from the NP-hard problem \textsc{(0,1,1,0)-pm}, which is the problem of deciding whether an instance $I$ of the popular matching problem with strict preferences contains a popular matching $M$, that does not contain a given forbidden edge $(x,y)$ and covers a given forced vertex $t\ne x,y$. By the reduction of \cite{faenza2018two}, we may assume that $x$ is a leaf, and $y$ has only one other neighbor $z$ (other than $x$), such that $y$ and $z$ are each others' first ranked agents. 
    
Let $I$ be an instance of \textsc{(0,1,1,0)-pm}.
    We create an instance $I'$ of \spopex\ as follows. We keep each vertex and edge from $I$. Furthermore, we add vertices $d_x,x',d_t,t'$ along with edges $(x,d_x),(x',d_x),(t,d_t),(t',d_t)$. Then, the graph is still clearly bipartite. 

    We let the preference of $d_x$ be $x\succ x'$ and the preference of $d_t$ be $t\sim t'$, where $\sim$ denotes a tie. Furthermore, we extend $t$'s preference list by adding $d_t$ to the end and also $x$'s preference list from $y$ to $d_x\sim y$.

    \begin{claim}
        If there is a \spop\ matching in $I'$, then there is a popular matching in $I$, which excludes $(x,y)$ and covers $t$.
    \end{claim}
    \begin{claimproof}
        Let $M'$ be a \spop\ matching in $I'$. Let $M$ be the restriction of $M'$ to $I$.
        
        We first show that $t$ is covered in $M$. Suppose otherwise. Then, if $(t,d_t)\in M'$, we get that $M'\setminus \{ (t,d_t)\} \cup \{ (t',d_t)\}$ dominates $M'$, as $d_t$ and $t'$ vote with $-1$ and only $t$ votes with $+1$. Similarly, if $(t,d_t)\notin M'$, then $M'\setminus \{ (t',d_t)\} \cup \{ (t,d_t)\}$ dominates $M'$, contradiction.

        Next, we show that $(x,y)\notin M$. Suppose that $(x,y)\in M'$. Then $(x',d_x)\in M'$. Also, $M(z)\ne \emptyset$, because otherwise $M'\setminus \{ (x,y)\} \cup \{ (y,z)\}$ dominates $M'$. However,  $M'\setminus \{ (x',d_x),(x,y),M(z)\} \cup \{ (x,d_x),(y,z)\}$ dominates $M'$, because $x,y,z$ and $d_x$ all vote with $-1$ and only $x'$ and $z$'s partner in $M'$ vote with $+1$.

        These observations also imply that $(t',d_t)\in M'$ and $(x,d_x)\in M'$. 

        Finally, we show that $M$ is popular. Suppose that a matching $N$ dominates $M$. Create $N'$ from $N$ by adding $(t',d_t)$ and $(x,d_x)$, if $N(x)=\emptyset$ and $(x',d_x)$ otherwise. In both cases, the sum of the votes of $t',d_t,x',d_x$ is $0$. Also, no original agent may have a smaller vote, because apart from $x$, each of them has the same situation in $M'$ and $M$; and also $N'$ and $N$. For agent $x$, if $(x,d_x)\in N'$, then they vote with 0 in both $M$ versus $N$ in $I$ and  $M'$ versus $N'$ in $I'$. Otherwise, their vote is $-1$ in the former and also $-1$ in the latter case. Hence, we get that the sum of votes in the comparison of $M'$ and $N'$ is negative, contradicting $M'$ is \spop.  
    \end{claimproof}

    \begin{claim}
        If there is a popular matching $M$ in $I$ that excludes $(x,y)$ and covers $t$, then there is a \spop\ matching in $I'$.
    \end{claim}
    \begin{claimproof}
        Let $M$ be such a popular matching in $I$. Then, we can extend $M$ to $M'$ in $I'$ by adding the edges $(x,d_x)$ and $(t',d_t)$. Suppose that some matching $N'$ dominates $M'$. 

        Let $N$ be the restriction of $N'$ to $I$. Each original agent other than $x$ and $t$ must have the same partners in both $M$ and $M'$ and also in $N$ and $N'$, hence their vote is the same. 

        Consider agent $t$. They have the same partner in $M$ and $M'$. Suppose that they have a different partner in $N$ than in $N'$. Then, $(t,d_t)\in N'$ and $t$ is unmatched in $N$. As $d_t$ is the strictly worst neighbor of $t$ in $I'$, we get that $t$ votes with $+1$ in both instances. 

        Consider agent $x$. Suppose that $x$ is matched in $N$, so $(x,y)\in N$ and $(x,y)\in N'$. As $x$ was unmatched in $M$ and $(x,d_x)\in M'$, we get that $x$ votes with $-1$ in both instances. If $x$ is not matched in $N$, but is matched in $N'$, then $(x,d_x)\in N'$ and $x$ votes with $0$ in both cases. Finally, if $x$ is unmatched in both $N$ and $N'$, then in $I$ they vote with $0$ in $M$ versus $N$, but they vote with $+1$ in $M'$ versus $N'$, which is larger. 

        It is easy to check that the sum of the votes of $x',d_x,t',d_t$ is always 0. Hence, we obtain that $N$ must dominate $M$, which is a contradiction. 
    \end{claimproof}
\end{proof}

\section{Experimental results for weak popularity}
\label{sec:exp}

Our main motivation for studying new notions for popular matchings was two-fold: on the one hand, popular matchings have desirable properties and are well-studied in the area of (computational) social choice, but have the drawback that -- contrary to weakly stable (and our new notion of weakly and $\gamma$-popular) matchings -- they might not exist when preferences contain ties. On the other hand, while weakly stable matchings always exist, they might leave a significant number of agents unmatched, and even then, finding one that maximises the number of matches is NP-hard \cite{IMMMmaxsmti}. We already highlighted that weakly popular matchings might match strictly more agents than weakly stable matchings, and it follows directly from Theorem \ref{thm:stableispop} that a maximum-size weakly popular matching is always at least as large as a maximum-size weakly stable matching. Now, we will complement these theoretical results with extensive experimental results to show that weakly popular matchings do indeed bring desirable size benefits, and to prove experimentally that, in practice, our approximation algorithm in Section \ref{sec:approxalgo} performs very close to optimally. \nd{We note that we did not perform experiments on super popular matchings due to the absence of efficient algorithms to find one, even if one exists.}

\paragraph{Setup and code}
All implementations were written in Python.
We used four compute nodes at a time, each equipped with dual AMD EPYC 7643 CPUs and 2TB RAM. 
All code is publicly available, see reference  \cite{experimentsCode}. 
The experiments involving randomly generated instances are seeded and can thus be easily replicated. All ILP models are solved using the PuLP package with the GurobiPy solver. We set a timeout of three hours for each algorithm execution per instance, and fewer than 1\% of executions of the ILP-based algorithms did not compute an optimal solution within this time.  We excluded these executions from our results, and we are convinced that the only impact this omission has on our analysis and interpretation below is that one should expect the true average runtimes of the ILP-based algorithms to be slower when computing to optimality.

\paragraph{Approximate and benchmark algorithms}
In our experiments, we compare five different algorithms, specifically, the
\begin{itemize}
    \item $\frac{2}{3}$-approx. algorithm {\sf MWS-A} for finding a max-size weakly stable matching by \citet{yokoi2021approximation};
    \item $\frac{3}{4}$-approx. algorithm {\sf MWP-A} for finding a max-size weakly popular matching from Section \ref{sec:approxalgo};
    \item ILP-based exact algorithm {\sf MWS-E} for finding a max-size weakly stable matching;
    \item ILP-based exact algorithm {\sf MWP-E} for finding a max-size weakly popular matching from Section \ref{sec:exactalgo};
    \item max-size matching algorithm {\sf MC} provided by the NetworkX Python package.
\end{itemize}
Stability verification is done in a straightforward linear-time way by iterating through the preference lists of one side of the bipartition once and checking for blocking edges. Weak popularity verification is done in polynomial time too by breaking ties such that matches appear last and then using the well-known characterisation due to \citet{huang2013popular}, similar to the procedure we described in Section \ref{sec:verify} for verifying the more general $\gamma$-popularity.

\paragraph{Experimental and statistical variables}
There are many possible variables to control and measure in the generation and consideration of instances. We characterise our instance spaces through the following independent and controlled variables:
\begin{itemize}
    \item $n_1,n_2$: numbers of agents on each side of the bipartition;
    \item $t_1,t_2$: tie density in the preference list on each side of the bipartition;\footnote{A tie density of 30\%, for example, implies that, given two adjacent preference list entries, there is a 30\% chance that the entries are tied. Multiple tied adjacent preference list entries form one larger tie. A tie density of 0 indicates strict preferences, and a tie density of 1 indicates that every preference list consists of one tie.}
    \item $l_1,l_2$: preference list lengths on each side of the bipartition.
\end{itemize}
For each instance and algorithm, we measure the following dependent variables:
\begin{itemize}
    \item size: cardinality of the computed matching;
    \item stable: boolean indicator of whether the computed matching is (weakly) stable or not;
    \item $ba$: total number of blocking agents (0 if stable);
    \item $bp$: total number of blocking edges (0 if stable);
    \item time: seconds required for the algorithm to terminate with an (optimal) solution;
\end{itemize}
Note that we do not include the time required to generate the instance in the time measurement, but we do include the time required for the creation of auxiliary instances, edge copies, matching projections, and ILP models, as these are specific to the algorithms.

\subsection{Experiments on random instances}

For our random experiments, we generated instances with preference lists sampled uniformly at random, a common practice in algorithmic experiments in this area (e.g., see \cite{mertens15,delorme2019mathematical,glitzner2025perspectives}). For symmetry and increased interpretability, we require that $n_1=n_2=n_*$ (for some $n_*\in\mathbb N$) and that $t_1=t_2=t_*$ (for some $t_*\in[0,1]$). Furthermore, we fix $l_1=l_*\cdot n_1$ (for some $l_*\in(0,1]$), which results in varying $l_2$ list lengths depending on the preferences being generated. For all averaged results, we generated 3,000 instances per setting to average over, resulting in more than 1M algorithm executions altogether.

We know by definition and by Theorem \ref{thm:stableispop} that the results must, in all settings, maintain the following hierarchy: {\sf MC} produces the largest matching, followed by {\sf MWP-E}, {\sf MWS-E}, and then {\sf MWS-A}. We also know that {\sf MWP-E} must (weakly) exceed {\sf MWP-A}, but the relative placement of our new approximation algorithm {\sf MWP-A} above {\sf MWS-E} is not necessarily guaranteed (see Theorem \ref{thm:approxguarantees} for guarantees). It is our goal to convince the reader that, experimentally, this is indeed the case and that our approximation algorithm achieves much tighter approximation bounds in practice than those we proved theoretically.

We start by studying the effects of $t_*$ and $l_*$ on the average size of matchings returned by the four algorithms, compared against {\sf MC} as a baseline. First, we fix $n_*=100$ and $l_*\in\{0.03, 0.3\}$, and study the impact of $t_*$ between 0\% and 100\% in steps of 10\% (with additional data points in steps of 5 for $l_*=0.3$ and $t_*\leq 0.6$). Figure \ref{fig:sizeversusties} summarises our results, showing the average size of matchings as a percentage of the average size of the maximum-cardinality matchings. Notice that the $y$-axis scaling differs between these two plots. This is because we can see much stronger absolute effects for $l_*=0.03$ than for $l_*=0.3$ (but the relative effect remains the same). We can observe that {\sf MWP-A} places above {\sf MWP-E} for all values of $t_*$, and that the average fraction of the {\sf MC} matching size is weakly increasing proportionally to $t_*$ for all other algorithms. For $l_*=0.03$, while the exact and approximate weakly stable matchings may leave more than 20\% of the number of agents matched in {\sf MC} unmatched,\footnote{The plot clearly shows a more than ten percentage point difference in the number of matches, and each match contains two agents.} the weakly popular algorithms achieve at least 98.12\% of the size of an {\sf MC} matching on average, for all values of $t_*$. For $l_*=0.3$, we see a similar effect, although both weakly stable algorithms achieve at least 99.85\% of the best-possible size, whereas both weakly popular matchings achieve 100\% for all $t_*$. 

\begin{figure}[!htb]
    \centering

    % Subplot 1
    \begin{subfigure}{0.44\textwidth}
        \begin{tikzpicture}
            \begin{axis}[
                width=\textwidth,
                height=5cm,
                ymin=88,
                ymax=100.5,
                xmin=-1,
                xmax=101,
                xlabel={$t_*$ in \%},
                grid=both,
                grid style={dashed, gray!30},
                cycle list name=color list,
                every axis plot/.append style={thick},
                ylabel={size ratio (\%)},
                title={Relative average sizes when $l_*=3\%$},
                title style={
                    yshift=-1.5ex
                }
            ]\addplot[acmDarkBlue, mark=*] table [x=t, y=MCr] {data/sizeties.txt};
            \addplot[acmGreen, mark=square*] table [x=t, y=MWPEr] {data/sizeties.txt};
            \addplot[acmPink, mark=triangle*] table [x=t, y=MWPAr] {data/sizeties.txt};
            \addplot[acmOrange, mark=diamond*] table [x=t, y=MWSEr] {data/sizeties.txt};
            \addplot[acmLightBlue, mark=pentagon*] table [x=t, y=MWSAr] {data/sizeties.txt};
            \end{axis}
        \end{tikzpicture}
    \end{subfigure}
    \hfill
    % Subplot 2
    \begin{subfigure}{0.55\textwidth}
        \begin{tikzpicture}
            \begin{axis}[
                width=.8\textwidth,
                height=5cm,
                ymin=99.83,
                ymax=100.01,
                xmin=-1,
                xmax=101,
                xlabel={$t_*$ in \%},
                grid=both,
                grid style={dashed, gray!30},
                cycle list name=color list,
                every axis plot/.append style={thick},
                title={Relative average sizes when $l_*=30\%$},
                legend style={
                    at={(1.05,0.5)},
                    anchor=west,
                },
                title style={
                    yshift=-1.5ex
                },
            ]\addplot[acmDarkBlue, mark=*] table [x=t, y=MCr] {data/sizeties2.txt};
            \addplot[acmGreen, mark=square*] table [x=t, y=MWPEr] {data/sizeties2.txt};
            \addplot[acmPink, mark=triangle*] table [x=t, y=MWPAr] {data/sizeties2.txt};
            \addplot[acmOrange, mark=diamond*] table [x=t, y=MWSEr] {data/sizeties2.txt};
            \addplot[acmLightBlue, mark=pentagon*] table [x=t, y=MWSAr] {data/sizeties2.txt};
            \addlegendentry{{\sf MC}}
            \addlegendentry{{\sf MWP-E}}
            \addlegendentry{{\sf MWP-A}}
            \addlegendentry{{\sf MWS-E}}
            \addlegendentry{{\sf MWS-A}}
            \end{axis}
        \end{tikzpicture}
    \end{subfigure}
    \caption{Effect of ties ($t_*$) on the matching size ratios for $n_*=100$}
    \label{fig:sizeversusties}
    \Description{The size ratio is increasing for all algorithms, but the popular matching algorithms achieve better ratios.}
\end{figure}

Next, we fix $n_*=100$ and $t_*\in\{0.03, 0.3\}$, and we study the impact of $l_*$ between 1\% and 100\% in steps of 1\% between 1\% and 10\%, and in steps of 10\% from thereon. Figure \ref{fig:sizeversuslists} summarises our results, showing only those for $l_*\leq40\%$, as all algorithms achieve an average of 100\% for $l_*\geq40\%$. Again, as expected from our analysis above, we find stronger effects for smaller values of $t_*$ and $l_*$. Specifically, we can see that {\sf MWP-E} and {\sf MWP-A} achieve much larger sizes on average than {\sf MWS-E} and {\sf MWS-A} whenever the latter match only a significantly lower number of agents. Also, the average size of matchings computed by {\sf MWP-A} always (weakly) exceeds that of {\sf MWS-E}. For both values for $t_*$ and for all algorithms, there seems to be a global minimum around $l_*=0.03$, i.e., when one side ranks exactly three agents of the other side, and our results suggest that this holds regardless of the probability of ties.

\begin{figure}[!htb]
    \centering

    % Subplot 1
    \begin{subfigure}{0.44\textwidth}
        \begin{tikzpicture}
            \begin{axis}[
                width=\textwidth,
                height=5cm,
                ymin=89,
                ymax=100.3,
                xmin=0,
                xmax=40,
                xlabel={$l_*$ in \%},
                grid=both,
                grid style={dashed, gray!30},
                cycle list name=color list,
                every axis plot/.append style={thick},
                ylabel={size ratio (\%)},
                title={Relative average sizes when $t_*=3\%$},
                title style={
                    yshift=-1.5ex
                },
            ]\addplot[acmDarkBlue, mark=*] table [x=l, y=MCr] {data/sizelists.txt};
            \addplot[acmGreen, mark=square*] table [x=l, y=MWPEr] {data/sizelists.txt};
            \addplot[acmPink, mark=triangle*] table [x=l, y=MWPAr] {data/sizelists.txt};
            \addplot[acmOrange, mark=diamond*] table [x=l, y=MWSEr] {data/sizelists.txt};
            \addplot[acmLightBlue, mark=pentagon*] table [x=l, y=MWSAr] {data/sizelists.txt};
            \end{axis}
        \end{tikzpicture}
    \end{subfigure}
    \hfill
    % Subplot 2
    \begin{subfigure}{0.55\textwidth}
        \begin{tikzpicture}
            \begin{axis}[
                width=.8\textwidth,
                height=5cm,
                ymin=92.5,
                ymax=100.2,
                xmin=0,
                xmax=40,
                xlabel={$l_*$ in \%},
                grid=both,
                grid style={dashed, gray!30},
                cycle list name=color list,
                every axis plot/.append style={thick},
                title={Relative average sizes when $t_*=30\%$},
                legend style={
                    at={(1.05,0.5)},
                    anchor=west,
                },
                title style={
                    yshift=-1.5ex
                },
            ]\addplot[acmDarkBlue, mark=*] table [x=l, y=MCr] {data/sizelists2.txt};
            \addplot[acmGreen, mark=square*] table [x=l, y=MWPEr] {data/sizelists2.txt};
            \addplot[acmPink, mark=triangle*] table [x=l, y=MWPAr] {data/sizelists2.txt};
            \addplot[acmOrange, mark=diamond*] table [x=l, y=MWSEr] {data/sizelists2.txt};
            \addplot[acmLightBlue, mark=pentagon*] table [x=l, y=MWSAr] {data/sizelists2.txt};
            \addlegendentry{{\sf MC}}
            \addlegendentry{{\sf MWP-E}}
            \addlegendentry{{\sf MWP-A}}
            \addlegendentry{{\sf MWS-E}}
            \addlegendentry{{\sf MWS-A}}
            \end{axis}
        \end{tikzpicture}
    \end{subfigure}

    \caption{Effect of list lengths ($l_*$) on the matching size ratios for $n_*=100$}
    \label{fig:sizeversuslists}
    \Description{The size ratio is increasing for all algorithms once $l_*$ crosses $3\%$, but the popular matching algorithms achieve better ratios.}
\end{figure}

Lastly, we fix the parameters $t_*=l_*=0.1$ and study the effect of $n_*$, chosen to be between 10 and 100 in steps of 1. Figure \ref{fig:sizeversusagents} summarises our results. We, again, find that the two algorithms returning weakly stable matchings can leave a significant fraction of agents unmatched that could otherwise be matched when relaxing weak stability to weak popularity. We can also see that {\sf MWP-A} has a performance that is extremely close to {\sf MWP-E}, despite our weaker theoretical guarantees, and, notably, {\sf MWP-A} always performs better than {\sf MWS-E}. The approximation gap between {\sf MWP-A} and {\sf MWP-E} also seems to be significantly tighter than that between {\sf MWS-A} and {\sf MWS-E}. Interestingly, compared to {\sf MC}, all other algorithms seem to exhibit minima in their achieved size ratios for $30\leq n_*\leq39$, which, given that $l_*=0.1$, again suggests a global minimum in performance whenever agents of one set rank exactly three agents of the other set. Finally, we emphasise that, in this setting, on average, {\sf MWP-A} matches fewer than one agent fewer than {\sf MC}, even at the observed minima.\footnote{At $n_*=30$, {\sf MC} matches at most 60 agents and {\sf MWP-A} achieves around 98.616\% the size of the {\sf MC} matching, which results in fewer than one unmatched agent in total. The same holds for all other values of $n_*$ studied here.} Finally, we remark that the performance appears approximately stepwise-constant in intervals of ten, which is likely because preference list lengths increase with every increase of ten for $n_*$ (due to $l_*=0.1$). This shows that adding (the same number of) agents to both sides of the market, while remaining the same constant preference list length (rather than one depending on the number of agents) on one side of the market, and the same probability of ties, does not affect the performance of the algorithms.

\begin{figure}[!htb]
    \centering
    \begin{tikzpicture}
        \begin{axis}[
            width=.99\textwidth,
            height=5cm,
            ylabel={size ratio (\%)},
            ymin=87,
            ymax=100.5,
            xmin=9,
            xlabel={$n_*$ agents per side},
            xmax=101,
            grid=both,
            grid style={dashed, gray!30},
            cycle list name=color list,
            every axis plot/.append style={thick},
            title={Relative average sizes when $l_*=t_*=0.1$},
            title style={
                yshift=-1.5ex
            },
            legend cell align={left},
            legend style={
                at={(1,0)},
                anchor=south east,
                column sep=1ex,
                font=\footnotesize
            },
        ]
        \addplot[acmDarkBlue, mark=*] table [x=n, y=MCr] {data/sizeagents.txt};
        \addplot[acmGreen, mark=square*] table [x=n, y=MWPEr] {data/sizeagents.txt};
        \addplot[acmPink, mark=triangle*] table [x=n, y=MWPAr] {data/sizeagents.txt};
        \addplot[acmOrange, mark=diamond*] table [x=n, y=MWSEr] {data/sizeagents.txt};
        \addplot[acmLightBlue, mark=pentagon*] table [x=n, y=MWSAr] {data/sizeagents.txt};
        \addlegendentry{{\sf MC}}
        \addlegendentry{{\sf MWP-E}}
        \addlegendentry{{\sf MWP-A}}
        \addlegendentry{{\sf MWS-E}}
        \addlegendentry{{\sf MWS-A}}
        \end{axis}
    \end{tikzpicture}
    \caption{Effect of agent numbers ($n_*$) on the matching size ratios for $l_*=t_*=0.1$}
    \label{fig:sizeversusagents}
    \Description{Popular matchings achieve a very high size ratio for all numbers of agents, while stable matchings can leave a significant number of agents unmatched.}
\end{figure}

Of course, while we established in Theorem \ref{thm:stableispop} that weakly stable matchings must be weakly popular, weakly popular matchings might not be weakly stable. Given the significant differences between the average maximum sizes of the two types of matchings that we highlighted in our experiments above, it can be expected that the matchings we compute using {\sf MWP-E} and {\sf MWP-A} can differ significantly from stable matchings. Hence, there is a trade-off between size and stability. An easy and natural way to quantify the instability of these weakly popular matchings is to consider the average numbers of blocking agents and blocking edges admitted by them (see \citet{eriksson2008instability}). Our results for this across different settings ($l_*,t_*\in\{0.2,0.8\}$ and $n_*\in\{50,100,150\}$) are summarised in Table \ref{table:blocking}. We can observe that, while maximum-size weakly popular matchings can deviate significantly from stable matchings, our approximation algorithm computes matchings that are, on average, very close to stable. Especially, it is clear that when the gap between maximum-size weakly stable and maximum-size weakly popular matchings is small (i.e., when $l_*,t_*$ are large, as established in our experiments above), {\sf MWP-A} essentially only computes stable matchings. Recall that $n_*$ is the number of agents \emph{per side}, so, in the worst case (in absolute numbers), when $n_*=50$ and $l_*=t_*=0.2$, {\sf MWP-A} computes a matching with, on average, 7.45 (out of $2n_*=100$ total agents) blocking agents and 5.67 blocking edges (out of $l_*\cdot n_*^2=0.2\cdot 50^2=250$ total edges). This corresponds to $7.45\%$ blocking agents and less than $2.27\%$ blocking edges. These observations render our approximation algorithm even more desirable than the exact ILP-based algorithm whenever matchings close to stable are desired, as we have already shown experimentally that the average gap in size between the matchings produced by the algorithms is minimal.

\begin{table}[!htb]
    \small
    \centering
    \renewcommand{\arraystretch}{.9}
    \caption{Average number of blocking agents/edges across different settings}
    \begin{tabular}{c c c c c c c c}
        \toprule
        & & \multicolumn{3}{c}{$t_*=20\%$} & \multicolumn{3}{c}{$t_*=80\%$} \\
        \cmidrule(lr){3-5} \cmidrule(lr){6-8}
        
        &  & $\mathbf{n_*}=$ \textbf{50} & \textbf{100} & \textbf{150} & \textbf{50} & \textbf{100} & \textbf{150} \\
        \midrule
        \multirow{2}{*}{$l_*=20\%$} 
        & {\sf MWP-A} & 7.45/5.67 & 3.51/2.81 & 1.50/1.23 & 0.02/0.01 & 0.00/0.00 & 0.00/0.00 \\
        & {\sf MWP-E} & 11.58/8.81 & 4.98/3.89 & 86.71/126.90 & 5.71/3.41 & 30.40/22.99 & 67.32/68.67 \\
        \midrule
        \multirow{2}{*}{$l_*=80\%$}
        & {\sf MWP-A} & 0.01/0.01 & 0.00/0.00 & 0.00/0.00 & 0.00/0.00 & 0.00/0.00 & 0.00/0.00 \\
        & {\sf MWP-E} & 20.77/32.99 & 24.58/55.25 & 19.67/58.39 & 25.56/35.59 & 67.99/134.73 & 102.96/268.50 \\
        \bottomrule
    \end{tabular}
    \label{table:blocking}
\end{table}

Finally, we study the average time required for the algorithms to terminate. Table \ref{table:times} summarises our results for the same settings as above ($l_*,t_*\in\{0.2,0.8\}$ and $n_*\in\{50,100,150\}$). Our main goal here is to convince the reader that the approximation algorithm {\sf MWP-A} can be implemented to run efficiently, after already having shown above that its performance is expected to be very close to that of the much slower optimal algorithm {\sf MWP-E} and exceeding that of the also much slower optimal algorithm {\sf MWS-E}. In the table, we can see that, as can be expected from its design, {\sf MWP-A} runs slower than {\sf MWS-A}, but significantly faster than both ILP-based algorithms. Furthermore, naturally, both approximation algorithms scale much better with an increase in agent numbers than the ILP-based algorithms.\footnote{Works such as those by \citet{delorme2019mathematical,delorme2021stability,pettersson2021improving} have shown that advanced preprocessing steps and non-trivial problem modelling or operations research techniques can reduce the time required to find a maximum-size weakly stable matching using ILPs. While we did not consider such optimisations, the ILP-based algorithms would, \nd{in the worst case,} still run in exponential time in the number of agents.  \nd{Thus, even if we were to incorporate such optimisations, the approximation algorithms would remain significantly more efficient, asymptotically speaking.}}

\begin{table}[!htb]
    \centering
    \caption{Average algorithm times across different settings in seconds}
    \begin{tabular}{c c c c c c c c}
        \toprule
        & & \multicolumn{3}{c}{$t_*=20\%$} & \multicolumn{3}{c}{$t_*=80\%$} \\
        \cmidrule(lr){3-5} \cmidrule(lr){6-8}
        
        &  & $\mathbf{n}=$\textbf{50} & \textbf{100} & \textbf{150} & \textbf{50} & \textbf{100} & \textbf{150} \\
        \midrule
        \multirow{4}{*}{$l_*=20\%$} 
        & {\sf MWS-A} & 0.00 & 0.03 & 0.07 & 0.00 & 0.03 & 0.06 \\
        & {\sf MWP-A} & 0.02 & 0.08 & 0.26 & 0.01 & 0.08 & 0.26 \\
        & {\sf MWS-E} & 0.09 & 0.36 & 1.46 & 0.11 & 0.67 & 2.81 \\
        & {\sf MWP-E} & 13.32 & 2.12 & 44.89 & 0.34 & 3.04 & 17.23 \\
        \midrule
        \multirow{4}{*}{$l_*=80\%$} 
        & {\sf MWS-A} & 0.04 & 0.21 & 0.73 & 0.04 & 0.21 & 0.72 \\
        & {\sf MWP-A} & 0.12 & 0.77 & 2.48 & 0.11 & 0.76 & 2.45 \\
        & {\sf MWS-E} & 0.60 & 9.69 & 58.19 & 0.88 & 12.97 & 172.29 \\
        & {\sf MWP-E} & 5.94 & 128.00 & 1259.23 & 4.01 & 53.42 & 924.08 \\
        \bottomrule
    \end{tabular}
    \label{table:times}
\end{table}

\subsection{Experiments on real-world instances}

To further verify the advantages of weak popularity and our approximation algorithm that we observed in the random experiments above, we also ran experiments on real-world instances. Specifically, we used datasets from the Scottish Foundation Allocation Scheme (SFAS), which used to be tasked with the assignment of medical school graduates in Scotland to so-called foundation training programs, which are required to practice medicine in the UK.\footnote{See \url{https://matching-in-practice.com/the-scottish-foundation-allocation-scheme-sfas/} for more details. The scheme has since been superseded by a UK-wide allocation.} Our datasets contain the full (anonymous) graduate preferences over foundation programs and the foundation program preferences over graduates, from the years 2005/06, 06/07, and 07/08. Graduate preferences are strict \new{(i.e., $t_1=0$)}, while the program preferences contain many large ties. \new{Specifically, for the years (in order),} we have $t_2=0.92,0.76,0.81$. 

As a first step, we converted the many-to-one matching instances into one-to-one matching instances by duplicating programs (this is a standard procedure and yields a one-to-one correspondence between the stable matchings, see for example \citet{manlove2013algorithmics}). We then applied the {\sf MC}, {\sf MWP-A}, and {\sf MWS-A} algorithms to the converted instances. We also took previously existing results for {\sf MWS-E} from \citet{kwanmanlove}, but did not run {\sf MWP-E} due to the sizes of the instances. Our results are shown in Table \ref{table:sfas}. We can see that our approximation algorithm {\sf MWP-A} returns a maximum-cardinality matching in two out of three instances and always strictly exceeds the matching size of the matching returned by {\sf MWS-E}. Notably, {\sf MWS-A} leaves up to 6\% of graduates unassigned. Of course, the trade-off is in the stability guarantees -- {\sf MWP-A} incurs a significant number of blocking agents and blocking edges, although much fewer than {\sf MC}. Note that blocking agents counts each hospital copy separately here, so the true number of blocking hospitals (rather than blocking hospital copies) is much lower. We highlight, furthermore, that for {\sf MWP-A}, the number of blocking residents (which are a subset of the blocking agents overall) is 232, 423, and 127 for years 05/06, 06/07, and 07/08, respectively. The times show that while {\sf MWP-A} is slower than {\sf MWS-A} (and much slower than {\sf MC}), all algorithms terminate in less than half a minute for each of the instances.

\begin{table*}[!htb]
    \centering
    \renewcommand{\arraystretch}{.9}
    \caption{SFAS results. Programs is the total number of program copies in the converted instances.}
    \begin{tabular}{c c c c c c}
        \toprule        
        &  & Matching Size & Blocking Agents & Blocking Edges & Time (s) \\
        \midrule
        Year 2005/06 & {\sf MC} & 759 & 1391 & 15410 & 0.72 \\
        (759 graduates & {\sf MWP-A} & 759 & 254 & 591 & 9.00 \\
        801 programs) & {\sf MWS-E} & 758 & 0 & 0 & N/A \\
         & {\sf MWS-A} & 748 & 0 & 0 & 1.92 \\
        \midrule
        Year 2006/07 &{\sf MC} & 781 & 1394 & 16177 & 0.72 \\
        (781 graduates & {\sf MWP-A} & 781 & 686 & 4454 & 26.98 \\
        789 programs) & {\sf MWS-E} & 746 & 0 & 0 & N/A \\
         & {\sf MWS-A} & 738 & 0 & 0 & 2.18 \\
        \midrule
        Year 2007/08 &{\sf MC} & 745 & 1325 & 14549 & 0.85 \\
        (748 graduates & {\sf MWP-A} & 739 & 520 & 1906 & 13.57 \\
        752 programs) & {\sf MWS-E} & 709 & 0 & 0 & N/A \\
         & {\sf MWS-A} & 702 & 0 & 0 & 2.01 \\
        \bottomrule
    \end{tabular}
    \label{table:sfas}
\end{table*}

\subsection{Finding large almost-stable weakly popular matchings}

While we could clearly demonstrate in our experiments that weakly popular matchings are desirable alternatives to weakly stable matchings when the size of the matching is crucial, we also detected significant instability among these matchings (particularly for {\sf MWP-E} and much less so for {\sf MWP-A}). However, at this point, it is not clear whether this is a necessary trade-off when optimising for size and weak popularity rather than stability. One piece of evidence against this is that, in the random experiments, we saw that for some parameter settings the maximum-size weakly stable and maximum-size weakly popular solutions have, on average, the same size, suggesting that no trade-off is necessary and that full stability is achievable. In an attempt to mitigate this instability issue as much as possible without sacrifices in size or weak popularity, we developed an extension of {\sf MWP-E}, which we will denote by {\sf MWP-MinBP}, to find a matching with a minimum number of blocking edges (or a slight variation that minimises blocking agents) among all maximum-size weakly popular matchings.

Let $I$ be an instance of \wpopOpt\ \nd{consisting of a simple graph (although we note below how to adjust this for parallel edges)} with edges $E$ and let $I'$ be the instance with edge duplications $E'$ derived from $I$ as described in Section \ref{sec:exactalgo}. Consider the following integer linear program (where $a(e)$ and $b(e)$ correspond to the $a$ and $b$ copies in $E'$ of any edge $e\in E$).

\[
\begin{aligned}
\text{maximize}\quad & (\vert U\vert\cdot\vert W\vert+1)\cdot\sum_{e\in E'} x_e - \sum_{e\in E} b_e \\
\text{subject to}\quad
& \sum_{e\in E'(v)} x_e \le 1
    \qquad\forall v\in U\cup W
    &&\\[6pt]
& \sum_{\substack{f\in E'(u)\\ e \,\ntriangleleft_u, f}} x_f
  + \sum_{\substack{f\in E'(w)\\ e \,\ntriangleleft_w, f}} x_f
  + x_e \ge 1
    \qquad\forall e=(u,w)\in E' \\[6pt]
    & \sum_{\substack{f\in E(u)\\ e \,\ntriangleleft_u, f}} (x_{a(f)}+x_{b(f)})
  + \sum_{\substack{f\in E(w)\\ e \,\ntriangleleft_w, f}} (x_{a(f)}+x_{b(f)})
  + x_{a(e)}+ x_{b(e)} + b_e\ge 1
    \;\;\;\forall e=(u,w)\in E \\[6pt]
& x_e\in\{0,1\}\qquad\forall e\in E'\\
& b_e\in\{0,1\}\qquad\forall e\in E.
\end{aligned}
\]

Notice that the only differences of this program compared to that for finding a maximum $\gamma$-popular matching are the additional constraints involving the $b_e$ variables, and a change of the objective function from maximize $\sum_{e\in E'} x_e$ to maximize $(\vert U\vert\cdot\vert W\vert+1)\cdot\sum_{e\in E'} x_e - \sum_{e\in E} b_e$ instead. Thus, clearly any $\gamma$-popular matching is a feasible matching for this program (e.g., set $b_e=1$ for all $e\in E$), and any feasible solution to this program is a $\gamma$-popular matching (by ignoring the $b_e$ variables). Now, furthermore, any optimal solution to this program is a maximum-size weakly popular matching, as 
$$\sum_{e\in E} b_e\leq \vert U\vert\cdot\vert W\vert<\vert U\vert\cdot\vert W\vert+1,$$
so maximizing the size, i.e., maximizing $\sum_{e\in E'} x_e$, remains the primary objective. Subject to finding a maximum-size weakly popular matching, we clearly minimise the number of blocking edges, as any maximum-size weakly popular matching with fewer blocking edges would also satisfy the constraints but have a larger objective function value. 

We also highlight that the program can be modified in a variety of ways. For example, \nd{if $I$ does not consist of a simple graph, the objective function must be changed by replacing $\vert U\vert\cdot \vert W\vert$ with $\vert E\vert$. W}e can also modify the program to find matchings with a minimum number of blocking agents instead: instead of defining the binary variable $b_e$ for all $e\in E$, we can define the binary variable $b_a$ for all $a\in U\cup W$. Then replace $b_e$ by $b_u+b_v$ in the constraint and impose the objective function maximize 
$$(\vert U\vert+\vert W\vert+1)\cdot\sum_{e\in E'} x_e - \sum_{a\in U\cup W} b_a$$ 
instead. By a similar argument as above, and noting that 
$$\sum_{a\in U\cup W} b_a\leq \vert U\vert+\vert W\vert<\vert U\vert+\vert W\vert+1,$$
This integer linear program characterises the set of maximum-size weakly popular matchings with a minimum number of blocking agents. Note that the formulation can be easily adjusted to only minimise blocking agents in $U$ or in $W$ instead, for example.

To investigate whether our initial observation turns out to be correct and there do indeed exist maximum-size weakly stable matchings that have significantly fewer blocking edges, we applied this modified ILP-based algorithm to the same instances analysed in Table \ref{table:blocking}. \nd{The comparison of {\sf MWP-MinBP} against {\sf MWP-E} is shown in Table \ref{table:blockingminbp}. Recall that both algorithms return maximum-size weakly popular matchings. }We can see in the table that {\sf MWP-MinBP} brings a significant advantage with respect to the average numbers of blocking edges: there are only few parameter settings where the average number is larger than 1 and only one setting where it is 2. {\sf MWP-E}, in contrast, returns matchings that contain up to 268.5 blocking edges per matching on average (in the worst setting). However, we note that there is a significant trade-off in computation time. For $n_*=50$ and $l_*=0.2$, {\sf MWP-MinBP} took only 0.48s (when $t_*=0.2$) or 0.75s (when $t_*=0.8$) per instance on average, but this increased steeply to 32.67s and 29.72s for $n_*=100$ and to 331.99s and 155.65s for $n_*=150$. For $l_*=0.8$, the times are much longer, starting at 16.21s ($t_*=0.2$) and 16.00s ($t_*=0.8$) per instance on average for $n_*=50$, increasing to 448.40s and 481.01s for $n_*=100$ and to 5038.24s and 5066.89s for $n_*=150$. We note, though, that while for $n_*\in\{50,100\}$ all computations terminated successfully, for $n_*=150$, 7.2\% (for $t_*=0.2$) and 37.3\% (for $t_*=0.8$) of instances did not finish within the 3h time limit, so the results are excluded from these calculations, as well as the summary statistics in Table \ref{table:blockingminbp}.

Hence, although {\sf MWP-MinBP} shows that \nd{maximising the size of weakly popular matchings need not come at a high cost of instability, the specific algorithm {\sf MWP-MinBP}} does not pose a feasible alternative when the number of agents is large, and the preference lists are long.

\begin{table}[!htb]
    \centering
    \renewcommand{\arraystretch}{.85}
    \caption{Comparison of blocking edges between ILP models}
    \begin{tabular}{c c c c c c c c}
        \toprule
        & & \multicolumn{3}{c}{$t_*=20\%$} & \multicolumn{3}{c}{$t_*=80\%$} \\
        \cmidrule(lr){3-5} \cmidrule(lr){6-8}
        
        &  & $\mathbf{n_*}=$ \textbf{50} & \textbf{100} & \textbf{150} & \textbf{50} & \textbf{100} & \textbf{150} \\
        \midrule
        \multirow{2}{*}{$l_*=20\%$}  
        & {\sf MWP-MinBP} & 1.93 & 1.42 & 0.43 & 0.50 & 0.82 & 2.00 \\
        & {\sf MWP-E} & 8.81 & 3.89 & 126.90 & 3.41 & 22.99 & 68.67 \\
        \midrule
        \multirow{2}{*}{$l_*=80\%$}  
        & {\sf MWP-MinBP} & 0.04 & 0.02 & 0.01 & 0.18 & 0.03 & 0.79 \\
        & {\sf MWP-E} & 32.99 & 55.25 & 59.39 & 35.59 & 134.73 & 268.50 \\
        \bottomrule
    \end{tabular}
    \label{table:blockingminbp}
\end{table}

\section{Conclusion and future work}
\label{sec:conclusion}

In this paper, we introduced \new{three} new natural notions of popularity in two-sided matching markets with ties. We have shown that a weakly popular and more generally a \cpop\ matching always exists, and we provided an algorithm that finds a $\frac{3}{4}$-approximation to the maximum-size \cpop\ matching problem. We also proved that the ratio of $\frac{3}{4}$ is tight assuming that either the strong-UGC or the Small Set Expansion Hypothesis holds. Furthermore, our extensive experimental evaluation on synthetic and real-world data confirmed the practical utility of weak popularity, showing that \new{weakly popular matchings can offer a significant advantage in matching size over stable matchings,} while maintaining a \new{high level of stability}. For super popularity, we have shown that even the existence problem is NP-hard. \nd{However, our proof only works for two-sided ties, so we leave the complexity with one-sided ties as an open question. It would also be nice to have a non-brute-force (exponential-time) algorithm deciding the existence of super popular matchings.}

% In the interest of anonymisation, please do not include acknowledgements in your submission.
%
%\begin{acks}
%
%
%\end{acks}

% Bibliography
\bibliographystyle{ACM-Reference-Format}
\bibliography{new-bibliography}

@String{Computing = "Computing" }

@String{Computer = "{IEEE} Computer" }

@String{Springer = "Springer-Verlag" }

@article{abraham2007popular,
  author       = {Abraham, David J. and Irving, Robert W. and Kavitha, Telikepalli and Mehlhorn, Kurt},
  title        = {Popular matchings},
  journal      = {SIAM Journal on Computing},
  volume       = {37},
  number       = {4},
  pages        = {1030--1045},
  year         = {2007},
  doi          = {10.1137/06067328X}
}

@article{chen2021matchings,
  author       = {Chen, Jiehua and Skowron, Piotr and Sorge, Manuel},
  title        = {Matchings under preferences: Strength of stability and tradeoffs},
  journal      = {ACM Transactions on Economics and Computation},
  volume       = {9},
  number       = {4},
  pages        = {20:1--20:55},
  year         = {2021},
  doi          = {10.1145/3485000}
}

@misc{nasre2017popular,
  title={Popular matching with lower quotas},
  author={Nasre, Meghana and Nimbhorkar, Prajakta},
  eprint        = {1704.07546},
  archivePrefix = {arXiv},
  primaryClass  = {cs.DS},
  year={2017}
}

@misc{csaji2023simple,
  author        = {Cs{\'a}ji, Gergely},
  title         = {A Simple 1.5-Approximation Algorithm for a Wide Range of Max-SMTI Problems},
  year          = {2023},
  eprint        = {2304.02558},
  archivePrefix = {arXiv},
  primaryClass  = {cs.DS},
  url           = {https://arxiv.org/abs/2304.02558}
}

@inproceedings{csaji2023approximation,
  author       = {Cs\'aji, Gergely and Kir\'aly, Tam\'as and Yokoi, Yu},
  title        = {Approximation algorithms for matroidal and cardinal generalizations of stable matching},
  booktitle    = {Symposium on Simplicity in Algorithms (SOSA)},
  pages        = {103--113},
    address = {Florence, Italy},
  publisher    = {SIAM},
  year         = {2023}
}

@incollection{cseh2017popular,
  author       = {Cseh, {\'A}gnes},
  title        = {Popular matchings},
  booktitle    = {Trends in Computational Social Choice},
  editor       = {Endriss, Ulle},
  publisher    = {AI Access Foundation},
  address      = {Amsterdam, The Netherlands},
  year         = {2017},
  pages        = {105--122},
}

@article{cseh2017onesidedties,
  author       = {{Cseh}, {\'A}gnes and Huang, Chien-Chung and Kavitha, Telikepalli},
  title        = {Popular matchings with two-sided preferences and one-sided ties},
  journal      = {SIAM Journal on Discrete Mathematics},
  volume       = {31},
  number       = {4},
  pages        = {2348--2377},
  year         = {2017},
  doi          = {10.1137/16M1076162}
}

@article{cseh2018dominant,
  author       = {{Cseh}, {\'A}gnes and Huang, Chien-Chung and Kavitha, Telikepalli},
  title        = {Popular edges and dominant matchings},
  journal      = {Mathematical Programming},
  volume       = {172},
  number       = {1},
  pages        = {209--229},
  year         = {2018},
  doi          = {10.1007/s10107-017-1183-y}
}

@article{CsehKavitha21,
  author       = {{Cseh}, {\'A}gnes and Kavitha, Telikepalli},
  title        = {Popular matchings in complete graphs},
  journal      = {Algorithmica},
  volume       = {83},
  number       = {5},
  pages        = {1493--1523},
  year         = {2021},
  doi          = {10.1007/s00453-020-00791-7}
}

@article{delorme2019mathematical,
  author = {Maxence Delorme and Sergio García and Jacek Gondzio and Jörg Kalcsics and David Manlove and William Pettersson},
  title        = {Mathematical models for stable matching problems with ties and incomplete lists},
  journal      = {European Journal of Operational Research},
  volume       = {277},
  number       = {2},
  pages        = {426--441},
  year         = {2019},
  doi          = {10.1016/j.ejor.2019.03.017}
}

@article{delorme2021stability,
    author = {Maxence Delorme and Sergio García and Jacek Gondzio and Joerg Kalcsics and David Manlove and William Pettersson},
  title        = {Stability in the hospitals/residents problem with couples and ties: Mathematical models and computational studies},
  journal      = {Omega},
  volume       = {103},
  pages        = {102386},
  year         = {2021},
  doi          = {10.1016/j.omega.2020.102386}
}

@incollection{smti1.5inapprox,
  author="Dudycz, Szymon
and Manurangsi, Pasin
and Marcinkowski, Jan",
  title        = {Tight inapproximability of minimum maximal matching on bipartite graphs and related problems},
  booktitle    = {Approximation and Online Algorithms, WAOA 2021},
    location ={Lisbon, Portugal},
  pages        = {48--64},
  publisher    = {Springer},
  year         = {2022},
    address={Lisbon, Portugal},
doi ={10.1007/978-3-030-92702-8_4},
}

@article{eriksson2008instability,
  author       = {Eriksson, Kimmo and H\"aggstr\"om, Olle},
  title        = {Instability of matchings in decentralized markets with various preference structures},
  journal      = {International Journal of Game Theory},
  volume       = {36},
  pages        = {409--420},
  year         = {2008},
  doi          = {10.1007/s00182-007-0081-6}
}

@misc{faenza2018two,
  author        = {Faenza, Yuri and Powers, Vladlena and Zhang, Xingyu},
  title         = {Two-Sided Popular Matchings in Bipartite Graphs with Forbidden/Forced Elements and Weights},
  year          = {2018},
  eprint        = {1803.01478},
  archivePrefix = {arXiv},
  primaryClass  = {cs.DM},
  url           = {https://arxiv.org/abs/1803.01478}
}

@inproceedings{faenza2019popular,
  author       = {Faenza, Yuri and Kavitha, Telikepalli and Powers, Vladlena and Zhang, Xingyu},
  title        = {Popular matchings and limits to tractability},
  booktitle    = {Proceedings of SODA 2019},
  pages        = {2790--2809},
  publisher    = {SIAM},
  year         = {2019},
doi={10.1137/1.9781611975482.173},
    address = {San Diego, California, USA},
}

@article{gale1962college,
  author       = {Gale, David and Shapley, Lloyd S.},
  title        = {College admissions and the stability of marriage},
  journal      = {The American Mathematical Monthly},
  volume       = {69},
  number       = {1},
  pages        = {9--15},
  year         = {1962},
  doi          = {10.2307/2312726}
}

@article{gardenfors1975match,
  author       = {G{\"a}rdenfors, Peter},
  title        = {Match making: assignments based on bilateral preferences},
  journal      = {Behavioral Science},
  volume       = {20},
  number       = {3},
  pages        = {166--173},
  year         = {1975},
  doi          = {10.1002/bs.3830200304}
}

@misc{glitzner2025perspectives,
  author        = {Glitzner, Frederik and Manlove, David},
  title         = {Perspectives on Unsolvability in the Roommates Problem},
  year          = {2025},
  eprint        = {2505.06717},
  archivePrefix = {arXiv},
  primaryClass  = {cs.GT},
  url           = {https://arxiv.org/abs/2505.06717}
}

@article{gupta2021popular,
  author  = {Gupta, Sushmita and Misra, Pranabendu and Saurabh, Saket and Zehavi, Meirav},
  title   = {Popular matching in roommates setting is {NP}-hard},
  journal = {ACM Transactions on Computation Theory},
  volume  = {13},
  number  = {2},
  pages   = {8:1--8:20},
  year    = {2021},
  doi     = {10.1145/3442354}
}

@incollection{halldorsson2002inapproximability,
  author    = {Halld{\'o}rsson, Magn{\'u}s M. and Iwama, Kazuo and Miyazaki, Shuichi and Morita, Yasufumi},
  title     = {Inapproximability results on stable marriage problems},
  booktitle = {LATIN 2002: Theoretical Informatics},
  series    = {Lecture Notes in Computer Science},
  volume    = {2286},
  pages     = {554--568},
  publisher = {Springer},
  year      = {2002},
    doi={10.1007/3-540-45995-2_48},
    address={Cancun, Mexico},
}

@software{experimentsCode,
  author       = {Glitzner, Frederik},
  title        = {Weakly-Popular and Super-Popular Matching Experiments},
  month        = 1,
  year         = 2026,
  publisher    = {Zenodo},
  doi          = {10.5281/zenodo.18299079},
    note ={Zenodo code repository}
}

@article{biro_sm_10,
   author = {Péter Biró and David Manlove and Shubham Mittal},
   doi = {10.1016/J.TCS.2010.02.003},
   issn = {0304-3975},
   issue = {16-18},
   journal = {Theoretical Computer Science},
   month = {3},
   pages = {1828-1841},
   publisher = {Elsevier},
   title = {Size versus stability in the marriage problem},
   volume = {411},
   year = {2010},
}

@InProceedings{gupta2020parameterized,
  author =	{Gupta, Sushmita and Jain, Pallavi and Roy, Sanjukta and Saurabh, Saket and Zehavi, Meirav},
  title =	{{On the (Parameterized) Complexity of Almost Stable Marriage}},
  booktitle =	{Proceedings of FSTTCS 2020},
  pages =	{24:1--24:17},
  ISBN =	{978-3-95977-174-0},
  ISSN =	{1868-8969},
  year =	{2020},
  volume =	{182},
  publisher =	{Leibniz-Zentrum f{\"u}r Informatik},
  address =	{Dagstuhl, Germany},
  doi =		{10.4230/LIPIcs.FSTTCS.2020.24},
}

@article{hungarian,
author = {Kuhn, Harold W.},
title = {The Hungarian method for the assignment problem},
journal = {Naval Research Logistics Quarterly},
volume = {2},
number = {1-2},
pages = {83-97},
doi = {10.1002/nav.3800020109},
year = {1955}
}

@misc{chen2025fptapproximabilitystablematchingproblems,
      title={FPT-Approximability of Stable Matching Problems}, 
      author={Jiehua Chen and Sanjukta Roy and Sofia Simola},
      year={2025},
      eprint={2508.10129},
      archivePrefix={arXiv},
      primaryClass={cs.GT},
      url={https://arxiv.org/abs/2508.10129}, 
}

@article{galesoto85,
title = {Some remarks on the stable matching problem},
journal = {Discrete Applied Mathematics},
volume = {11},
number = {3},
pages = {223-232},
year = {1985},
doi = {10.1016/0166-218X(85)90074-5},
author = {David Gale and Marilda Sotomayor},
}

@InProceedings{studentallocation,
author="Kwanashie, Augustine
and Irving, Robert W.
and Manlove, David F.
and Sng, Colin T. S.",
editor="Jan, Kratochv{\'i}l
and Miller, Mirka
and Froncek, Dalibor",
title="Profile-Based Optimal Matchings in the Student/Project Allocation Problem",
booktitle="Combinatorial Algorithms",
year="2015",
publisher="Springer",
address="Cham",
pages="213--225",
doi={10.1007/978-3-319-19315-1_19}
}

@inproceedings{refugeematching,
author = {Strasser Ceballos, Clara and Kern, Christoph},
title = {Location matching on shaky grounds: Re-evaluating algorithms for refugee allocation},
year = {2025},
publisher = {ACM},
address = {New York, USA},
doi = {10.1145/3715275.3732149},
booktitle = {Proceedings of the 2025 Conference on Fairness, Accountability, and Transparency},
pages = {2180–2199},
numpages = {20},
series = {FAccT '25}
}

@article{suksompong2021constraints,
  author       = {Warut Suksompong},
  title        = {Constraints in Fair Division},
  journal      = {ACM SIGecom Exchanges},
  volume       = {19},
  number       = {2},
  pages        = {46--61},
  year         = {2021},
  doi          = {10.1145/3505156.3505162},
}

@misc{minimax,
      title={A Minimax Perspective on Almost-Stable Matchings}, 
      author={Frederik Glitzner and David Manlove},
      year={2026},
      eprint={2601.14195},
      archivePrefix={arXiv},
      primaryClass={cs.GT},
      url={https://arxiv.org/abs/2601.14195}, 
}

@inproceedings{glitznermanloveasaamas26,
author = {Glitzner, Frederik and Manlove, David},
title = {Minimax and Preferential Almost-Stable Matchings},
year = {2026},
booktitle = {Proceedings of AAMAS 2026},
publisher={IFAAMAS},
address = {Paphos, Cyprus},
numpages = {9},
doi={10.65109/PCDE6577}
}

@article{hamada09,
   author = {Koki Hamada and Kazuo Iwama and Shuichi Miyazaki},
   doi = {10.1016/J.IPL.2009.06.008},
   issn = {0020-0190},
   issue = {18},
   journal = {Information Processing Letters},
   month = {8},
   pages = {1036-1040},
   publisher = {Elsevier},
   title = {An improved approximation lower bound for finding almost stable maximum matchings},
   volume = {109},
   year = {2009},
}

@incollection{halldorsson2003improved,
  author    = {Halld{\'o}rsson, Magn{\'u}s M. and Iwama, Kazuo and Miyazaki, Shuichi and Yanagisawa, Hiroki},
  title     = {Improved approximation of the stable marriage problem},
  booktitle = {Algorithms -- ESA 2003},
  series    = {Lecture Notes in Computer Science},
  volume    = {2832},
  pages     = {266--277},
  publisher = {Springer},
    address ={Budapest, Hungary},
  year      = {2003}
}

@article{huang2013popular,
  author  = {Huang, Chien-Chung and Kavitha, Telikepalli},
  title   = {Popular matchings in the stable marriage problem},
  journal = {Information and Computation},
  volume  = {222},
  pages   = {180--194},
  year    = {2013},
  doi     = {10.1016/j.ic.2012.10.012}
}

@article{irving1994stable,
  author  = {Irving, Robert W.},
  title   = {Stable marriage and indifference},
  journal = {Discrete Applied Mathematics},
  volume  = {48},
  number  = {3},
  pages   = {261--272},
  year    = {1994},
  doi     = {10.1016/0166-218X(92)00179-P}
}

@incollection{IMMMmaxsmti,
  author    = {Iwama, Kazuo and Manlove, David and Miyazaki, Shuichi and Morita, Yasufumi},
  title     = {Stable marriage with incomplete lists and ties},
  booktitle = {ICALP 1999},
  series    = {Lecture Notes in Computer Science},
  volume    = {1644},
  pages     = {443--452},
  publisher = {Springer},
  year      = {1999},
    address={Prague, Czech Republic},
}

@article{kamiyama2020popular,
  author  = {Kamiyama, Naoyuki},
  title   = {Popular matchings with two-sided preference lists and matroid constraints},
  journal = {Theoretical Computer Science},
  volume  = {809},
  pages   = {265--276},
  year    = {2020},
  doi     = {10.1016/j.tcs.2019.12.017}
}

@article{kavitha2014size,
  author  = {Kavitha, Telikepalli},
  title   = {A size-popularity tradeoff in the stable marriage problem},
  journal = {SIAM Journal on Computing},
  volume  = {43},
  number  = {1},
  pages   = {52--71},
  year    = {2014},
  doi     = {10.1137/120902562}
}

@inproceedings{kiraly2012linear,
  author    = {Kir{\'a}ly, Zolt{\'a}n},
  title     = {Linear time local approximation algorithm for maximum stable marriage},
  booktitle = {Proceedings of the Second International Workshop on Matching Under Preferences},
    address={Budapest, Hungary},
  year      = {2013},
    publisher={mdpi},
    pages={471-484},
    doi={10.3390/a6030471}
}

@article{csaji2025simple,
  title={A Simple 1.5-Approximation Algorithm for a Wide Range of Maximum-Size Stable Matching Problems},
  author={Cs{\'a}ji, Gergely and Kir{\'a}ly, Tam{\'a}s and Yokoi, Yu},
  journal={Mathematics of Operations Research},
  year={2025},
  publisher={INFORMS}
}

@incollection{kwanmanlove,
  author    = {Kwanashie, Augustine and Manlove, David},
  title     = {An integer programming approach to the hospitals/residents problem with ties},
  booktitle = {Operations Research Proceedings},
  pages     = {263--269},
  publisher = {Springer},
    address={Rome, Italy},
  year      = {2013}
}

@book{manlove2013algorithmics,
  author    = {Manlove, David},
  title     = {Algorithmics of Matching Under Preferences},
  publisher = {World Scientific},
  year      = {2013},
    address={Singapore},
  isbn      = {978-981-4450-14-0}
}

@article{manlove2002hard,
  author  = {Manlove, David and Irving, Robert W. and Iwama, Kazuo and Miyazaki, Shuichi and Morita, Yasufumi},
  title   = {Hard variants of stable marriage},
  journal = {Theoretical Computer Science},
  volume  = {276},
  number  = {1--2},
  pages   = {261--279},
  year    = {2002},
  doi     = {10.1016/S0304-3975(01)00206-7}
}

@article{mertens15,
doi = {10.1088/1742-5468/2015/01/P01020},
year = {2015},
month = {jan},
publisher = {IOP Publishing and SISSA},
volume = {2015},
number = {1},
pages = {P01020},
author = {Mertens, Stephan},
title = {Stable roommates problem with random preferences},
journal = {Journal of Statistical Mechanics: Theory and Experiment},}

@article{irving02srt,
  author  = {Irving, Robert W. and Manlove, David},
  title   = {The stable roommates problem with ties},
  journal = {Journal of Algorithms},
  volume  = {43},
  pages   = {85--105},
  year    = {2002},
  doi     = {10.1006/jagm.2002.1219}
}

@incollection{Mcdermid09,
  author    = {McDermid, Eric},
  title     = {A 3/2-approximation algorithm for general stable marriage},
  booktitle = {ICALP 2009},
  series    = {Lecture Notes in Computer Science},
  volume    = {5555},
  pages     = {689--700},
    address={Rhodes, Greece},
  publisher = {Springer},
  year      = {2009},
    doi={10.1007/978-3-642-02927-1_57}
}

@article{pettersson2021improving,
  author  = {Pettersson, William and Delorme, Maxence and Garc{\'\i}a, Sergio and Gondzio, Jacek and Kalcsics, Joerg and Manlove, David},
  title   = {Improving solution times for stable matching problems through preprocessing},
  journal = {Computers \& Operations Research},
  volume  = {128},
  pages   = {105128},
  year    = {2021},
  doi     = {10.1016/j.cor.2020.105128}
}

@phdthesis{yanagisawa2007approximation,
  author  = {Yanagisawa, Hiroki},
  title   = {Approximation algorithms for stable marriage problems},
  school  = {Kyoto University},
  year    = {2007}
}

@misc{yokoi2021approximation,
  author        = {Yokoi, Yu},
  title         = {An Approximation Algorithm for Maximum Stable Matching with Ties and Constraints},
  year          = {2021},
  eprint        = {2107.03076},
  archivePrefix = {arXiv},
  primaryClass  = {cs.DS},
  url           = {https://arxiv.org/abs/2107.03076}
}

@inproceedings{krishnapriya2018good,
  title={How Good Are Popular Matchings?},
  author={Krishnapriya, AM and Nasre, Meghana and Nimbhorkar, Prajakta and Rawat, Amit},
  booktitle={17th International Symposium on Experimental Algorithms (SEA 2018)},
  volume={103},
    address={L'Aquila, Italy},
  pages={9},
  year={2018},
  publisher={Schloss Dagstuhl--Leibniz-Zentrum fuer Informatik}
}

@inproceedings{nasre2017popularity,
  title={Popularity in the generalized hospital residents setting},
  author={Nasre, Meghana and Rawat, Amit},
    volume={10304},
  booktitle={International Computer Science Symposium in Russia},
  pages={245--259},
  year={2017},
    address={Kazan, Russia},
  publisher={Springer}
}

\end{document}